\documentclass[11pt]{article}
\usepackage[utf8]{inputenc}
\usepackage{enumitem}
\usepackage{amsmath,amssymb}
\usepackage{nicefrac}
\usepackage{fullpage}

\usepackage{tcolorbox}
\usepackage{amsthm}
\usepackage{mathtools}
\usepackage{amsfonts}
\usepackage{dsfont}
\usepackage{graphics}
\usepackage{xcolor}
\usepackage{dsfont}
\usepackage{thmtools}
\usepackage{thm-restate}

\usepackage{mdframed}
\definecolor{shadecolor}{gray}{0.6}

\usepackage[ruled,vlined,linesnumbered,algonl]{algorithm2e}
\SetEndCharOfAlgoLine{}
\SetKwComment{Comment}{\footnotesize$\triangleright$\ }{}

\SetCommentSty{mycommfont}

\usepackage[unicode=true,
bookmarks=false,
breaklinks=false,pdfborder={0 0 1},backref=none,colorlinks=true,allcolors=blue]
{hyperref}
\usepackage{cleveref}
\usepackage[normalem]{ulem}

\Crefname{algocf}{Algorithm}{Algorithms}
\crefname{algocfline}{line}{lines}
\Crefname{invariant}{Invariant}{Invariants}
\Crefname{claim}{Claim}{Claims}
\Crefname{subclaim}{Subclaim}{Subclaims}

\usepackage{setspace}

\makeatletter
\allowdisplaybreaks
\makeatother

\usepackage[compact]{titlesec}

\newcommand{\sketch}[1]{}

\newtheorem{thm}{Theorem}[section]

\newtheorem{lem}[thm]{Lemma}

\newtheorem{claim}[thm]{Claim}

\newtheorem{defn}[thm]{Definition}

\newtheorem{fact}[thm]{Fact}

\def \RR   {{\mathbb R}}

\newcommand*\mc[1]{\mathcal{#1}}

\newcommand{\cI}{\mathcal{I}}
\newcommand{\cM}{\mathcal{M}}
\newcommand{\cP}{\mathcal{P}}

\DeclarePairedDelimiter\abs{\lvert}{\rvert}

\DeclareMathOperator{\rank}{rank}

\DeclareMathOperator{\InvExp}{\text{InvExp}}
\DeclareMathOperator{\Span}{\text{span}}
\DeclareMathOperator{\circuit}{\text{circuit}}
\DeclareMathOperator{\Market}{\mathbb{M}}

\newcommand{\sse}{\subseteq}
\newcommand{\opt}[1]{#1^*}
\newcommand{\alg}[1]{\widehat{#1}}

\newcommand{\alert}[1]{{\color{red}#1}\marginpar{$\star\star$}}
\newcommand{\agnote}[1]{\textcolor{red}{AG:#1}\marginpar{$\star\star$}}

\definecolor{debianred}{rgb}{0.84, 0.04, 0.33}

\definecolor{purple}{rgb}{0.59, 0.44, 0.84}

\usepackage[textsize=tiny,textwidth=2cm,color=green!50!gray,obeyFinal]{todonotes}

\newcommand{\Ind}{I^*}
\newcommand{\Left}{L}
\newcommand{\Right}{R}

\date{}

\title{Maintaining Matroid Intersections Online}

\author{Niv Buchbinder\thanks{Dept.\ of Statistics and Operations Research, Tel Aviv University, Israel. Supported in part by Israel Science Foundation grant 2233/19 and United States - Israel Binational Science Foundation grant 2022418.} \and Anupam Gupta\thanks{Carnegie Mellon University and Google Research. This research was partly supported by NSF awards CCF-1955785, CCF-2006953, and CCF-2224718. } \and Daniel Hathcock\thanks{Department of Mathematical Sciences, Carnegie Mellon University, Pittsburgh PA 15213. Supported by the NSF Graduate Research Fellowship grant DGE-2140739.} \and Anna R. Karlin \thanks{University of Washington, Seattle, WA 98195. This research was supported in part by NSF award CCF-1813135 and Air Force Office of Scientific Research grant FA9550-20-1-0212.} \and Sherry Sarkar\thanks{Department of Mathematical Sciences, Carnegie Mellon University, Pittsburgh PA 15213.}}
\date{\today}

\begin{document}

\maketitle

\begin{abstract}

Maintaining a maximum bipartite matching online while minimizing recourse/augmentations is a well studied problem, motivated by content delivery, job scheduling, and hashing. A breakthrough result of Bernstein, Holm, and Rotenberg (\emph{SODA 2018}) resolved this problem up to a logarithmic factors. However, we may need a richer class of combinatorial constraints (e.g., matroid constraints) to model other problems in scheduling and resource allocation.

We consider the problem of maintaining a maximum independent set of an arbitrary matroid $\mathcal{M}$ and a partition matroid $\mathcal{P}$ in the online setting. Specifically, at each timestep $t$ one part $P_t$ of the partition matroid (i.e., a subset of elements) is revealed: we must now select at most one of these newly-revealed elements, but can exchange some of the previously selected elements for new ones from previous parts, to maintain a maximum independent set on the elements seen thus far. The goal is to minimize the number of augmentations/changes done by our algorithm. If $\mathcal{M}$ is also a partition matroid, we recover the problem of maintaining a maximum bipartite matching online with recourse as a special case. In our work, we allow arbitrary matroids $\mathcal{M}$, and so we can model broader classes of problems.

Our main result is an $O(n \log^2 n)$-competitive algorithm, where $n$ is the rank of the largest common base; this matches the current best quantitative bound for the bipartite matching special case. Our result builds substantively on the breakthrough result of Bernstein, Holm, and Rotenberg for maintaining bipartite matchings: a key contribution of our work is to make connections to market equilibria and prices, and our use of properties of these equilibria in submodular utility allocation markets to prove our bound on the number of augmentations.
\end{abstract}

\newpage
\setcounter{page}{1}

\section{Introduction}
\label{sec:introduction}

In the \emph{Online Matroid Intersection Maintenance Problem with
  recourse}, we want to maintain a maximum independent set in the
intersection of an arbitrary matroid $\mathcal{M}$ and a partition
matroid $\mathcal{P}$ in the online setting. Specifically, suppose the
partition matroid is given by a partition $(P_1, P_2, \ldots, P_\ell)$
of the element set $E$, and $\cM$ is another matroid on
$E$.~\footnote{A \emph{matroid} $\cM$ over a ground set $E$ is given
  by a downward-closed collection of independent sets
  $\cI \subseteq 2^E$ such that for any $A, B \in \cI$ with $|A|<|B|$
  there exists an element $e \in B\setminus A$ such that
  $A\cup\{e\} \in \cI$~\cite[\S39]{schrijver2003volB}.}  Both of these
matroids are initially unknown to us. Now at each timestep $t$, we
have a current maximum independent set $I_{t-1}$ and the next part
$P_t$ of the partition matroid is revealed.

Since we need to maintain a maximum independent set in the
intersection of $\mc{M}$ and the portion of the partition matroid seen
so far, we may need to perform some augmentations---i.e., we may need
to drop elements from $I_{t-1}$ and add elements from
$E\setminus I_{t-1}$ to the current independent
set. 
Our objective is to minimize the total number of reassignments (i.e.,
additions or deletions from the independent set) over the course of
the arrival of all parts of $\mc{P}$. 

A special case of our problem that has been considered
extensively~\cite{GroveKKV95,BosekLSZ14,BosekLSZ18,BosekLSZ22,BHR18}
is that of the online \emph{bipartite matching} problem with
recourse. This setting corresponds to the matroid $\mc{M}$ also being
a partition matroid. In turn, it allows us to identify the elements
with edges of a bipartite graph whose vertices are the parts in the
two partitions. Hence each timestep corresponds to a new vertex from one side of the graph
arriving, along with its incident edges. 

In order to maintain a maximum matching, we need to augment along
alternating paths, which corresponds to dropping or adding edges. (If
individual edges arrive one-by-one rather than vertices, nothing better than $\Omega(n^2)$
total cost is possible in the worst case. For example, in the instance where edges arrive on alternating ends of a path, the augmentation must be the entire path at each step~\cite[\S1]{BHR18}.)

In addition to its natural and combinatorial appeal, the generality of
the online matroid intersection maintenance problem allows us to model
problems in resource allocation and scheduling beyond the matching case: 
\begin{itemize}
\item \emph{Laminar matroids} generalize the bipartite matching
  setting to allowing constraints on a hierarchy of ``groups'': e.g.,
  suppose clients arrive online and need to be matched to a server
  from their desired subset. However, we may have restrictions on the
  number of clients assigned to servers on the same server rack, or
  the same data center (due to cooling, power, and bandwidth
  constraints). These can be captured by laminar matroids, where we
  are given capacities on a family of laminar sets. Laminar
  restrictions are common when considering such job scheduling
  problems with restricted (hierarchical) resources.
  
\item A different setting is that of \emph{matroid
    partitioning}~\cite{Edmonds65}: the elements of a single matroid $\cM$
  arrive over time, and need to be partitioned among $k$ color
  classes, so that each color class is an independent set in $\cM$.
  The goal is to minimize the number of color changes. In this
  setting, the underlying matroid constraint captures the scheduling
  constraints of a single server cluster (e.g., like in the laminar
  case above, or the examples below), the coloring captures the
  idea of partitioning the jobs among these clusters, and the recourse
  bound ensures that only a few jobs are reassigned between
  clusters. The matroid coloring problem can be modeled using online
  matroid intersection by the simple idea of ``lifting'' elements to
  (element, color) pairs. 

\item \emph{Transversal matroids} are a useful matroid constraint for
  modern schedulers, since they can model
  \emph{coflows}~\cite{ChowdhuryS12} (which are tasks that can be
  jointly processed on a computing resource, e.g., can be shuffled in
  parallel on a MapReduce cluster~\cite{ImMPP19}).  In transversal
  matroids the elements are \emph{nodes} on one side of a bipartite
  graph, and independent sets correspond to matchable subsets of
  nodes. Combining these with the matroid partitioning idea allows us to
  partition a collection of jobs among
  clusters. In particular, if we now have several different computing resources (clusters), and jobs arriving online, the scheduler algorithm needs to decide which cluster to choose for each job in order to process it, and the goal is to minimize the number of jobs that have to be switched between clusters.

\item If we have a routing problem instead of a scheduling
  one~\cite{JahanjouKR17}, we can use \emph{gammoids} instead of
  transversal matroids: these capture sets of nodes which admit
  vertex-disjoint flows to a sink; again the goal would be to
  minimize reroutings.
\end{itemize}

\subsection{Our Result and Techniques}
\label{sec:main-result}

The augmenting-path algorithm for matroid
intersection~\cite{AignerDowling71,Lawler75} (see also
\cite[\S41.2]{schrijver2003volB}) immediately bounds the total number
of reassignments by $O(n^2)$, where $n$ is the rank of the maximum
independent set in $\cP \cap \cM$. To our knowledge, nothing better
was known about this problem in general prior to the current
paper. Our main result is the
following: 

\begin{restatable}[Main Theorem]{thm}{MainThm}
  \label{thm:main-result}
  The \emph{Shortest Augmenting Path} (SAP) algorithm for the
  \emph{Online Matroid Intersection Maintenance} problem results in at
  most $O(n \log^2 n)$ total reassignments, where $n$ is the rank of the
  intersection $\mc{M} \cap \cP$. 
\end{restatable}

Our analysis reveals a perhaps surprising connection between matroid
intersection maintenance and the well-known theory of market
equilibrium prices. Indeed, we define a market with the arriving parts
viewed as buyers who want to get one item/element from their desired subset.
Viewing the items as being divisible allows us to find a market equilibrium,
where the price of each item gives us crucial information about the
length of the shortest augmenting path starting with that item. Note that the
algorithm is purely combinatorial; the market helps us expose the
properties of the underlying instance, and to argue about the
algorithm's performance.

The bound of $O(n \log^2 n)$ matches that given by the breakthrough
paper of Bernstein, Holm and Rotenberg~\cite{BHR18} for the special
case of online matchings; this is not a coincidence. Indeed, our first
step is to reinterpret their work in the language of market
equilibria; we then develop the machinery and connections to reason
about general matroid intersection markets and obtain our results. We now elaborate
on these connections and our techniques.

\paragraph{Bipartite Matching} In the special case of online bipartite matching, we imagine the
vertices (``clients'') of one side of a bipartite graph, along with
the induced edges to the other side (``servers''), arriving online.
\cite{BHR18} define a notion of \emph{server necessity} to capture how
much of each server $s$ is needed to match all clients. They compute
server necessities via ``balanced flows'' (which is the solution to a
certain convex program), and also via an intuitive combinatorial
decomposition (also called a ``matching skeleton'' in other works) 
that builds on Hall's Theorem. These ingredients give an {\em expansion
  lemma} bounding the length of the shortest augmenting path for a new
client in terms of the minimum server necessity among its
neighbors.

We begin our investigation by showing in \S\ref{sec:matchings} how the
language and machinery of market equilibria yield a concise and
appealing version---though unchanged in its fundamentals---of
Bernstein et al's proof for the online bipartite matching problem with
recourse.  Our starting point is a reinterpretation of the concept of
server necessity as \emph{prices in a market equilibrium}. A Fisher
market~\cite{Brainard-Scarf} consists of $n$ buyers and $m$ divisible
items: each buyer $i$ arrives with a budget of money $m_i$ and a
utility function that specifies buyers' utilities for each possible
bundle of goods. A \emph{market equilibrium} is a set of prices
$p_1, \ldots, p_m$, where $p_j$ is the price of item $j$, such that
each buyer spends their money on a utility-maximizing bundle, the
supply precisely equals the demand, and the market clears (i.e., each
good with positive price is sold and each buyer spends all their
money). If we view clients in the online matching problem as buyers
each having one dollar, and the servers as items, with each buyer
having equal and linear utility for each item in their bipartite graph
neighborhood (and zero utility for non-neighbors), the market clearing
prices turn out to be precisely the server necessities. The
equilibrium allocation and prices can be computed using the
Eisenberg-Gale convex program~\cite{EG59} (which differs from the
convex program used in~\cite{BHR18}).

\paragraph{Matroid Intersection} We then show how this market equilibrium perspective allows us to
generalize to online \emph{matroid intersection} with recourse, where
we have a more general set of feasibility conditions on the
allocations. Again, we start with $n$ buyers and $m$ items, but the
items are now the elements $E$ of matroid $\mc{M}$, and the buyers are
interested in disjoint elements. (In the online matching problem,
these elements of the matroid are edges incident to the buyer/online
vertex, and hence map to a set of offline vertices that buyer wants to
match to.)  Each buyer again arrives with some money $m_i$ and a
linear utility function over the items in its part. In this market,
the allocation of items to buyers must lie in the matroid polytope of
$\mc{M}$. This market again is an Eisenberg-Gale market~\cite{JV07}, for which a
market equilibrium exists and can be be computed with a

convex program of the form:
\[ \max \bigg\{ \sum_i m_i \log \sum_{e \in P_i} y_e \biggm| \sum_{e \in S} y_e \leq \rank_{\mc{M}}(S)
  ~~\forall S \subseteq E, y \geq 0 \bigg\}. \] The dual variables to
the program again yield equilibrium prices for each item.

Using these ideas, we extend the ideas from
\S\ref{sec:matchings} to the general online matroid intersection
problem in \S\ref{sec:general-matroids}.

Our setting requires new ideas beyond the case of matchings because
the convex program is richer: the prices (i.e., duals) are now on sets
and not on elements. There is a natural way to translate from sets to
elements: the price of each element is the sum of prices of sets
containing it---but then the prices are not unique, and we can no
longer argue the monotonicity of prices as new clients arrive, a crucial
ingredient in~\cite{BHR18} and in \S\ref{sec:matchings}.  To address this, 
we first make a connection to submodular utility allocation
markets~\cite{JV10} to show that prices seen by buyers are
monotone. Then we show a decomposition theorem for matroids (extending
such a result for matchings \cite{GoelKK12,BHR18}) that allows us to
define unique and consistent prices for elements, and to show
monotonicity of all individual element prices over arrivals.

In conjunction with this, we can define for any element $e$ a
collection of nested sets, showing that if there are no short
augmenting paths (in the natural exchange graph) starting at this
element, then these nested sets grow exponentially at rate
$\approx (1/p_e)$. With this, we bound the length of paths by
$\approx \frac{\ln n}{1-p_e}$ (of course, the paths are never of
length more than $n$.) Finally, the monotonicity of element prices allows
us to distribute this augmentation cost to the price increase of the
elements participating in this augmenting path.  Since the element
prices lie in $[0,1]$, the cost charged to each element (over the
entire run of the algorithm) is
\[ \int_{p = 0}^{1} \min \bigg\{ n, \frac{\ln n}{1-p} \bigg\} \, dp
  \leq \int_{p = 0}^{1-\nicefrac{1}{n}} \frac{\ln n}{1-p} \, dp + \int_{p = 1-\nicefrac{1}{n}}^{1} n \, dp = O(\ln^2
  n). \] 
  A technical detail is that summing this over all elements would
give us $|E| \log^2 n$, and not something that depends on the
rank. This issue can be handled using a convexity-based argument.

\paragraph{Paper Outline} In the remainder of the paper, we first
illustrate the basic ideas of our arguments in \S\ref{sec:matchings} for the setting of
bipartite matchings. In \S\ref{sec:general-matroids} we give details for the general case of maintaining matroid intersections. Finally, we close with some remarks and future directions in \S\ref{sec:closing}.

\subsection{Related work}

To our knowledge, online matroid intersection maintenance with
recourse has not been studied previously. The special case of online
bipartite matching problem with recourse was defined by
\cite{GroveKKV95}, who gave an $\Omega(n \log n)$ lower
bound. \cite{ChaudhuriDKL09} gave optimal algorithms with
$O(n \log n)$ recourse when clients arrive in random order, or when
the graph is a forest. For the case of forests,
\cite{BosekLSZ18,BosekLSZ22} studied the shortest augmenting path
algorithm and showed it to also be optimal. The first breakthrough on
the general case of matching maintenance was by \cite{BosekLSZ14} who gave a
$O(n^{1.5})$ recourse bound; eventually \cite{BHR18} gave the current
best $O(n \log^2 n)$ bound.

The problem of load-balancing with recourse is closely related:
\cite{AGZ99,PW98,Wes00} show how to allocate jobs to machines and
maintain near-optimum load while reassigning $O(\log n)$ jobs per
timestep. \cite{AzarBK94,AKPPW93,AzarNR92} show results for dynamic
settings without reassignments, and observe strong lower bounds. 
\cite{GKS14-matching} show how to allocate unit jobs to machines in a
restricted machines setting to maintain a load of $(1+\varepsilon)L$
with $O(1/\varepsilon)$ recourse; they give results for a dynamic flow
variant. Recently similar results were given by \cite{KLS22} for the
case of unrelated machines, with logarithmic recourse. Very recently \cite{BBLS23} studied a more general setting in which covering-packing constraints arrive and depart online and should be satisfied upon arrival. This setting captures as a special case a fully dynamic fractional load balancing/matching problem in which jobs arrive and depart online. They obtained an $O(\log (n/\varepsilon)/\varepsilon)$-competitive algorithm when the algorithm is given a $(1+\varepsilon)$ resource augmentation.

Several works \cite{BHK09, EFN23} have also modelled recourse in online matching with \textit{buybacks} or \textit{cancellations}. In these settings, there is instead a penalty for recourse; for every online vertex that is matched the algorithm earns money, but the algorithm may choose to ``buy back" resources from offline vertices, incurring a penalty. The buyback setting has also been extended beyond matchings to matroid and matroid intersection constraints \cite{AK09, BHK09, Varadaraja11}.

There is an enormous body of work on online bipartite matching
problems {\em without} recourse, starting with the seminal work of
Karp, Vazirani and Vazirani~\cite{KVV90}. In these settings, the
algorithm makes irrevocable decisions, and the goal is to maximize the
size/weight of the matchings; see, e.g.,~\cite{Mehta13,EIV23}. An
extension of this to matroid intersection was studied by
\cite{GuruganeshS17}, who considered two arbitrary matroids defined on
the same ground set whose elements arrive one at a time in a random
order, and must be irrevocably picked/discarded, to maximize the size
of the independent set selected. Another large body of work studies
the min-weight perfect matching problem (mostly in metric settings);
see,
e.g.,~\cite{MeyersonNP06,BBGN14,R18,PS21}. 
The techniques in these works are orthogonal to ours.

Our combinatorial decomposition for matroids produces a \emph{matroid intersection skeleton}
extending that for
matching; this decomposition for matching was studied by~\cite{BHR18},
and previously, under the name of \emph{matching skeletons}
by~\cite{GoelKK12,LeeS17} with the goal of understanding streaming algorithms
for matchings. The matching skeleton was also used to derive an optimal competitive ratio in the batch arrival model of online bipartite matching \cite{FN20, FNS21}. To the best of our knowledge, the extension to matroids
has not been studied before; making further connections to streaming
algorithms for matroid intersection remains an interesting future
direction. 

As discussed above, our work makes a connection to and builds on basic
results on market equilibria~\cite{AD54,EG59,Brainard-Scarf}. 
Market equilibria and especially the design of algorithms for computing these
equilibria have been the subject of intense study by the algorithmic game theory
community over the last two decades. For an introduction to the topic, see chapters 5 and 6 of ~\cite{NRTV07}. 
Of particular relevance to us is the 
work of Jain and Vazirani~\cite{JV07} on Eisenberg-Gale markets~\cite{EG59}.

Convex programming techniques, and in particular the Eisenberg-Gale ``fair allocation"
convex program have also been used to guide combinatorial algorithms
before, e.g., in the context of flow-time
scheduling~\cite{ImKM18,GGKS19}. However, these prior works do not
consider the cost of recourse; they use the convex program to directly
schedule jobs. We instead use it to compute prices and show the
existence of short augmenting paths.

Note that while the problem we study can viewed as a dynamic graph problem, 
the cost function we study (bounding recourse) is unrelated to the kinds of cost functions studied in the dynamic graph algorithms literature.

\section{Maintaining Matchings via Markets}
\label{sec:matchings}

We now present our market equilibria-based perspective for the
bipartite matching case; we build on this for general matroids in
\S\ref{sec:general-matroids}.  For matchings, the adversary fixes a
bipartite graph $(B,T,E)$ with $n$ buyers $B$ and $m$ items $T$.  The
vertices in $T$ (the offline side) are known up-front, whereas the
vertices in $B$ (and the edges between them) are revealed online (we can assume that the maximum matching after $i$ arrivals
  has size $i$, and hence the maximum matching has size $n = |B|$; this is without loss of generality, see~\cite[Obs.~9]{BHR18}). We
see the edges between the $i^{th}$ buyer (also called $i$) and its
neighbors $N(i)$ only at time $i$.

If $M_{i-1}$ is the maximum matching maintained by the algorithm after
seeing $i-1$ vertices, and buyer $i$ arrives, the \emph{shortest
  augmenting path} algorithm (SAP) finds an (arbitrary) shortest
augmenting path from $i$ to a free item (if such a path exists), and
augments the matching $M_{i-1}$ along this path to get $M_i$. 

Let $\ell_i$ denote the length of this shortest augmenting path found
by the algorithm, and the goal is to bound the worst-case value of
$\sum_{i=1}^n \ell_i$. There exist instances for which
$\sum_{i-1}^n \ell_i = \Omega (n \log n)$~\cite[Thm.~1]{GroveKKV95};
the following result of \cite{BHR18}---which we prove using market
equilibria in this section---matches this lower bound up to a logarithmic factor. 

\begin{restatable}{thm}{BHRThm}
  \label{thm:BHR}
  The \emph{Shortest Augmenting Path} (SAP) algorithm performs  $O(n \log^2 n)$  changes. 
\end{restatable}

\subsection{Preliminaries on the Fisher Market}
\label{sec:preliminaries}

In the \emph{Fisher's linear model}~\cite{Brainard-Scarf}, we have $n$
buyers and $m$ divisible goods. Each buyer $i$ has a \emph{budget}
$m_i$. The utility functions are linear: buyer $i$ derives a utility
$u_{ij} y_{ij}$ out of being allocated an amount $y_{ij}$ of good
$j$. 
There exist ``market-clearing'' prices for the goods and a
corresponding equilibrium allocation of the goods to buyers in which
(i) each good with a positive price is fully sold; (ii) buyers are only allocated goods
that maximize their utility-per-price, sometimes called
``bang-per-buck" (if a good has price 0, no buyer has positive
utility for it, so we assume that each item has an interested buyer); and (iii) each buyer
spends their entire budget.  

Eisenberg and Gale~\cite{EG59} showed how to compute the market
equilibrium allocation and prices in the Fisher model using the
following convex program.
\begin{equation}\tag{EG1}\label{eq:EG1}
\begin{aligned}
	\max \quad & \textstyle \sum_{i=1}^n m_i \ln \big(\sum_{j=1}^{m} u_{ij} y_{ij} \big) & \\
\text{s.t.,} \quad 	& \textstyle \sum_{i=1}^{n} y_{ij} \le 1&\forall j=1,2, \ldots, m\\
	& y_{ij}\geq 0. &
\end{aligned}	
\end{equation}

Let $\{y_{ij}\}$ be the optimal solution to the convex program. The Lagrangian dual variable $p_j$ for each item $j$ in the program can be interpreted as the price of that item. We imagine that the budget of a buyer when it arrives is 1 (and 0 before it arrives) and the utilities are $u_{ij} = 1$ for all edges $(i,j) \in E$ and $0$ otherwise.
The following properties can be derived from the KKT optimality conditions (see, e.g., \cite[Chapter~5]{NRTV07}). 

\begin{thm}
  \label{thm:ME}
  Let $\{p_j\}_{j=1}^m$ be an optimal dual solution to (\ref{eq:EG1}). Then the following hold:
  \begin{enumerate}[nosep,label=(\roman*)]
      \item All items are fully allocated:  $\sum_{i =1}^{n}y_{ij}=1$ for all $j$.  \label{prop1}  
 \item Buyers are buying the cheapest price items: $y_{ij}>0 \Rightarrow p_j = \min_{j' \in N(i)} p_{j'}$. \label{prop2}
 \item Each buyer $i$ spends all of their money: $\sum_{j=1}^{m}p_j \cdot y_{ij} = 1$, and hence  $\sum_{j=1}^{m}p_j = n$.  \label{prop3}
 \end{enumerate}
\end{thm}

Thus, with unit utilities each buyer $i$ buys only the lowest-price items
from $N(i)$; we refer to this lowest price for buyer $i$ as $q_i$. Let
us prove an additional important property of market-clearing prices
for this utility function.

\begin{restatable}{lem}{monotone}
  \label{lem:monotone}
  Let $i$ be the $i^{th}$ buyer to arrives and let $p, p'$ be the
  equilibrium prices before and after it is added. Let
  $q := \min_{j\in N(i)}\{p_j\}$. Then \vspace{-0.2cm}
  \begin{equation}
    \label{newprices}	
    \begin{array}{ll}
      p'_j =p_j & \text{if }p_j < q \\
      p'_j \ge p_j &  \text{if }p_j \ge q.
    \end{array}	 
  \end{equation}
\end{restatable}

\begin{proof}
 Let $B_{\geq q}$ (resp.\ $T_{\geq q}$) be the set of buyers that
  buy at price at least $q$ (resp.\ the set of items whose
  price is at least $q$) in the market equilibrium immediately
  prior to the arrival of buyer $i$. No buyer in $B_{\ge q}$ has an
  edge to an item in $T_{<q}$, since it would buy such an item
  otherwise. So if we find a market equilibrium for the subproblem
  consisting of buyers in $B_{\ge q}\cup \{i\}$ and show that all
  prices do not decrease in that equilibrium, we will be done: that
  equilibrium together with the equilibrium for $B_{< q}$ satisfies
  KKT conditions and hence is the new equilibrium.

Hence, let us restrict our attention to buyers in
  $B_{\ge q} \cup \{i\}$ and their neighbors $T_{\ge q}$. Towards a
  contradiction, let $T_< := \{j \mid p'_j<p_j\}$ be the subset of
  items whose price decreases in the new market equilibrium. Let
  $B_< := \{\ell \mid y_{\ell j}>0, j\in T_<\}$ be the buyers who buy
  items from $T_<$. Since each such buyer expends all their budget
  and all goods are completely sold, we have
  $\sum_{j\in T_<} p_j \leq |B_<|$. However, after the price update
  \emph{all} the minimum price items for buyers in $B_<$ must be in
  $T_<$. Moreover, these item prices strictly decreased, and other
  item prices may increase or remain the same. Thus, it must be that
  $\sum_{j\in T_<}p_j >\sum_{j\in T_<}p'_j \geq |B_<|$, where the
  second inequality holds since the buyers now spend all their budgets
  on items from $T_<$. This is a contradiction.
\end{proof}

\subsection{The Expansion Lemma}

\newcommand{\Match}{M^\star}

A \emph{tail augmenting path} is an alternating path that starts from
an arbitrary item, alternates between matched and unmatched edges, and
ends in a free item. The \emph{length} of a (tail) augmenting path is
the number of items on the path.  The following key lemma relates the
length of these paths to the item prices.

\begin{lem}[Expansion lemma]
  \label{lem:expansion}
  Let $(B,T,E)$ be a bipartite graph, let $\Match$ be an arbitrary
  matching, and let $p$ be the market clearing prices. Then, for any
  item $j^*$ with $p_{j^*}\in [0,1)$, there is a tail augmenting path from
  $j^*$ whose length is $O\big(\frac{\ln n}{1-p_{j^*}}\big)$.
\end{lem}

\begin{proof}
  Consider market clearing prices $p$ and an allocation
  $y_{ij}$. Let $j^*\in T$ be an arbitrary item with price
  $p_{j^*}\in[0,1)$. We will denote $\Match(i)$ to be the item that buyer $i$ is matched to under $\Match$. If $p_{j^*}=0$ then there is no $i$ such that
  $j^*\in N(i)$, the item is therefore unmatched, and the claim holds.
  Otherwise, we define the following sets inductively,\vspace{-0.2cm}
  \begin{align*}
    \Right_1 & =\{j^*\}\\
    \Left_k & = \{i \mid \Match(i) \in \Right_{k}\} & k=1, 2, \ldots, \\
    \Right_{k+1} &= \{j \mid \ y_{ij}>0, i\in \Left_{k}\} & k=1, 2,  \ldots,
  \end{align*}
  That is, $\Left_k$ is the set of buyers that are matched in $\Match$ to items
  in $\Right_{k}$, and $\Right_{k+1}$ is the set of items that are bought by
  some buyer in $\Left_k$ in the market clearing allocation. We prove the
  following inductively: If $\Right_1, \ldots, \Right_{k}$ do not contain an
  unmatched item, then (a)~$p_{j} \leq p_{j^*}$ for all items $j\in
  \Right_{k+1}$, and (b)~the size $|\Right_{k+1}| \geq (\nicefrac{1}{p_{j^*}})^{k}$.

  The base case for $\Right_1$ trivially holds. We begin by proving property (a) of the induction. Consider an item $j \in \Right_{k+1}$. Since it is bought strictly positively by some item $i \in \Left_{k}$, its price $p_j$ must be the price paid by $i$, that is, $p_j = q_i$. Moreover, observe that $q_i \leq p_{\Match(i)}$, by definition. And since $i \in \Left_{k}$, we have $\Match(i) \in \Right_{k}$, so the induction hypothesis implies $p_{\Match(i)} \leq p_{j^*}$. Together, this shows that $p_j =q_i \leq p_{\Match(i)} \leq p_{j^*}$ as desired.
    
  To prove property (b) of the induction, suppose that $\Right_{k}$ has no unmatched items, then we have that $|\Left_{k}|= |\Right_{k}|$. Hence, we have,
  \begin{align}
    |\Right_{k+1}| & \textstyle \triangleq |\{j \mid \ y_{ij}>0, i\in \Left_{k}\}| = \sum_{j \in \Right_{k+1}}\sum_{i : j\in N(i)}y_{ij} \label{exp-ineq1}\\
            & \textstyle\geq \sum_{i\in \Left_{k}} \sum_{j\in \Right_{k+1}}y_{ij} \label{exp-ineq2}\\
            & \textstyle= \sum_{i\in \Left_{k}}\frac{1}{q_i} \geq \frac{|\Left_{k}|}{p_{j^*}} =
              \frac{|\Right_{k}|}{p_{j^*}} \geq \big(\frac{1}{p_{j^*}}\big)^{k}. \label{exp-ineq3}
  \end{align}
  Equality \eqref{exp-ineq1} holds since every item $j$ that is allocated is
  fully sold, and so $\sum_{i : j\in N(i)} y_{ij}=1$, and
  \eqref{exp-ineq2} holds since the RHS sums only on the subset of
  $y_{ij}>0$ from $\Left_{k}$ to $\Right_{k+1}$.  Finally, \eqref{exp-ineq3}
  holds by \Cref{thm:ME}  \ref{prop2},\ref{prop3} since each buyer
  $i\in \Left_{k}$ spends all its money on items of price $q_i$ that are
  in $\Right_{k+1}$ and so $q_i \cdot \sum_{j\in \Right_{k+1}} y_{ij} = 1$. The next
  inequality holds since $q_i\leq p_{j^*}$, and finally we use the
  induction hypothesis.

  By our construction, if $\Right_{k}$ contains an unmatched item $j$, then
  there is a tail augmenting path from $j^*$ to $j$ of length $k$. To conclude the proof, note
  that for any $\Right_k$ whose items are all matched, $|\Right_k|\leq n$ (the
  number of buyers), since otherwise there must be an unmatched item
  in $\Right_k$. Thus, $(\frac{1}{p_{j^*}})^{k-1} \leq |\Right_k| \leq n$.
  Simplifying we get that the length of such a tail augmenting path is
  at most $O\big(\max\big\{1,\frac{\ln(n)}{\ln (1/p_{j^*})}\big\}\big)= O\big(\frac{\ln n}{1-p_{j^*}}\big)$, where the final inequality follows
  from $1-x \leq \ln \nicefrac{1}{x}$ for $x\in(0,1]$. 
\end{proof}

\subsection{Bounding the Augmentations}

\begin{fact} 
\label{fact:boundedprices}
The price of all items \emph{before the arrival} of the $i^{th}$ buyer is either $1$ or at most $1-\frac{1}{i}$. There is no tail augmenting path from items with price $1$.
 \end{fact}
 
\begin{proof}
  Recall that we are assuming that all buyers can be matched, hence all prices being at most $1$ follows from the KKT optimality conditions for \eqref{eq:EG1} and Hall's theorem. Let $T_p$ be a set of items with some
  price $p\leq1$, and let $B_p := \{i \mid y_{ij}>0, j\in T_p\}$. Then, since the buyers in $B_p$ buy only items in $T_p$, we have
  $|B_p|= \sum_{i\in B_p, j\in T_p} p \cdot y_{ij} = p \cdot |T_p|$, giving
  $p=\nicefrac{|B_p|}{|T_p|}$. This fraction is either $1$, or at most
  $\frac{i-1}{i}$ (because $|B_p|\leq i-1$ before the arrival of $i$). For the items of price $1$, the buyers in $B_1$ have edges only to items in $T_1$ (since otherwise, they would buy cheaper items) and $|B_1|=|T_1|$. Thus, by Hall's theorem, they are all matched and there can be no tail augmenting path from such items.
\end{proof}  

\BHRThm*

\begin{proof}

For each item $j$, let $p_{j}(i)$ be its price \emph{before} the arrival of the $i^{th}$ buyer. 
  Let $j_{\min}\in N(i)$ be a neighbor of $i$ of minimal price, and define
  $q_{\min}(i) := \min_{j\in N(i)}\{p_{j}(i)\}$. By \Cref{fact:boundedprices} and the assumption that each arriving buyer can be matched, $q_{\min}(i)\le 1-1/n$. 

  Let $\Delta p_{j}(i)$ be the change in the price of item $j$ due to the arrival of buyer $i$. By \Cref{lem:expansion} the length
  of the shortest augmenting path from $i$ is at most that of the tail
  augmenting path from $j_{\min}$ which can be bounded as follows,  
  \begin{align}
    \ell_i & = O\left(\frac{\ln n}{1-q_{\min}(i)}\right)
             =  O\left(\frac{\ln n}{1-q_{\min}(i)}\right) \sum_{j :
             p_{j}(i)\geq q_{\min}(i)}\Delta p_{j}(i) \leq O\left(\ln
             n\right)
             \cdot \sum_{j\in T}\frac{\Delta p_{j}(i)}{1-p_{j}(i)},  \label{main-ineq}
  \end{align}
  where the second equality uses \Cref{thm:ME} \ref{prop3} to infer
  that $\sum_{j=1}^{m}p_j$ equals the current number of matched
  buyers, and \Cref{lem:monotone} says that the price
  of items with $p_{j}<q_{\min}(i)$ do not change, which together
  imply that 

  $\sum_{j : p_{j}\geq q_{\min}(i)}\Delta p_{j}(i)=1$. 

Hence, the total number of augmentations is
    \[ \sum_{i=1}^{n}\ell_i \leq O(\ln n) \cdot
      \underbrace{\sum_i\sum_{j\in T}\frac{\Delta
          p_{j}(i)}{1-p_{j}(i)}}_{(\star)} \leq O(\ln n) \cdot n \cdot
      \left( 1+ \int_{p = 0}^{1 - \nicefrac{1}{n}} \frac{dp}{1-p} \right) = O(n
      \ln^2 n), \] where the second inequality uses two facts:
    firstly, because $\sum_{j\in T}\Delta p_j(i)=1$ and as
    $\nicefrac{1}{1-x}$ is monotonically increasing, the sum $(\star)$
    is maximized when $n$ items have a final value of $1$. Secondly,
    we bound the last term of any one sum
    $\sum_i\frac{\Delta p_{j}(i)}{1-p_{j}(i)}$ by $1$, and the 
    previous terms---for which $p_{j}(i)\leq 1 - \nicefrac{1}{n}$---by the
    integral (note that $\Delta p_{j}(i) \geq 0$ for all $i, j$ by \Cref{lem:monotone}). This completes the proof. 
\end{proof}

\section{Extending to a General Matroid}
\label{sec:general-matroids}

We now prove our main result, for the intersection of a partition
matroid $\cP$ with a general matroid $\cM$ over the same ground set
$E$. The matroid $\cP$ has $n$ parts, each with rank $1$.
The elements of the $i^{th}$ part are
revealed at time $i$ (the order of the parts is unknown). We assume that there is an
independent set of size $n$ after the final arrival\footnote{This is w.l.o.g.: if an arriving part does not cause an augmentation, then no future augmenting path will pass through that part. And thus, the part can be ignored henceforth. See also \cite[Obs. 9]{BHR18} for a description in the matching setting.}, and hence the
maximum independent set after the arrival of $i$ parts has size
$i$. 

The \emph{shortest augmenting path} also works in this setting.
Let $I$ be the chosen independent set
before the arrival of $i$, and $\mc{P}|_i$ and $\mc{M}|_i$ be the
matroids restricted to elements $E|_i$, the elements which have been revealed thus
far. Define the
\emph{exchange graph} $D_{\mc{P}|_i, \mc{M}|_i}(I)$ as the bipartite
graph on nodes $(I, E|_i \setminus I)$ with (directed) arcs:
\begin{enumerate}[nolistsep]
    \item $y \to x$ is an arc of $D_{\mc{P}|_i, \mc{M}|_i}(I)$ if $I - y + x$ is independent in $\mc{P}$.

    \item $y \gets x$ is an arc of $D_{\mc{P}|_i, \mc{M}|_i}(I)$ if $I - y + x$ is independent in $\mc{M}$. 
\end{enumerate}
The algorithm finds a shortest path in $D_{\mc{P}|_i, \mc{M}|_i}(I)$
from some element in $P_i$ to a \emph{free element} in $\mc{M}$
(that is, some $e \not \in I$ for which $I + e$ is independent in
$\mc{M}$). This is an \emph{augmenting path}: it defines a valid sequence of exchanges to form a
new independent set, and the resulting independent set will have size
$\abs{I} + 1$. The correctness of this algorithm and its analysis
in the offline setting are due to Aigner and Dowling, and also Lawler~\cite{AignerDowling71,Lawler75} (see also \cite[\S41.2]{schrijver2003volB}).

Again, let $\ell_i$ be the length of the shortest augmenting path upon
the arrival of the $i$th part, and we want to bound the worst-case
value of $\sum_{i=1}^n \ell_i$. We restate our main Theorem:

\MainThm*

The remainder of the paper is dedicated to the proof of \Cref{thm:main-result}. It proceeds
analogously to our proof for bipartite matchings in \S\ref{sec:matchings}: we begin by defining a corresponding market we call the {\em matroid intersection market} which, in conjunction with the matroid intersection skeleton, yields ``prices'' for the elements of $E$. We then use properties of these prices to prove an ``Expansion Lemma'' generalizing \Cref{lem:expansion}. This lemma bounds the length of augmenting paths in terms of prices, and and hence gives our main result. However, as
mentioned in \S\ref{sec:introduction}, each of these steps requires
us to build on the ideas we used for the matchings case. 

\subsection{The matroid intersection (MI) market}
\label{sec:MI-market}

The {\em matroid intersection (MI) market} is defined for an arbitrary
matroid $\mc{M} = (E, \cI)$ and a partition matroid $\mc{P}$ on the
elements in $E$. There is a set $B$ of \emph{buyers}, where each part
$P_i$, $i=1, \ldots, |B|$ is associated with a buyer $i$. In this
market, $E$ is the set of \emph{items}, the items in $P_i$ are precisely
those that are of interest to buyer $i$, and each buyer arrives at the
market with a budget of $m_i$ dollars. The utility $u_i$ to agent $i$
for allocation $\{y_{e}\}_{e \in P_i}$ is
$\sum_{e \in P_i} y_e$, and the allocation constraints are
$\sum_{e \in S} y_e \le \rank_{\mc{M}}(S)$ for each $S\subseteq E$.
It is immediate that the Fisher market with $u_{ij} \in \{0,1\}$ is
the special case in which $\mc{M}$ is a partition matroid. A market
equilibrium for an MI market can be computed using the following
convex program that optimizes over the matroid polytope of $\mc{M}$:

\begin{equation}\tag{EG2}\label{eq:EG2}
\begin{aligned}
    \max\quad &\textstyle \sum_{i \in B} m_i \cdot \log \big( \sum_{e \in P_i} y_{e} \big) \\
     \quad & \textstyle \sum_{e \in S} y_{e} \leq \rank_{\mc{M}}(S) \quad \forall S \subseteq E\quad \quad &&(\alpha_S) \\
    &y_{e} \geq 0. 
\end{aligned}
\end{equation}
Here $y_e$ is the amount of item $e$ allocated (to the buyer whose
part contains $e$), and $\alpha_S$ for $S \subseteq E$ are a set of
dual variables. The optimal dual is not necessarily unique.  For any
optimal dual solution $\alpha$, we define \emph{prices} as
\begin{gather}
  p_e := \sum_{S \ni e} \alpha_S \quad \text{for each element }e \in
  E. \label{eq:prices-defn}
\end{gather}
The KKT conditions for \eqref{eq:EG2} are the following:
\begin{enumerate}[label=(\Alph*),nosep]
    \item Stationarity and complementary slackness:
\begin{align}
p_e \ge \frac{1}{\sum_{e' \in P_i} y_{e'}} \quad\text{ and }\quad  y_e \cdot \left(p_e  -  \frac{1}{\sum_{e' \in P_i} y_{e'}}\right)&= 0 \quad\quad \forall e\in E. \label{eq:KKT1} \\
\alpha_S \cdot\left( \sum_{e' \in S} y_{e'} - \text{rank}_{\mc{M}}(S)\right) &= 0 \quad\quad \forall S \subseteq E.\quad\quad \label{eq:KKT2}&
\end{align}
\item Primal feasibility and $\alpha_S \geq 0$ for all $S \subseteq E$ (dual feasibility).
\end{enumerate}

\begin{lem}\label{thm:MI-market-conditions}
Let $\alpha_{S}$ be any optimal dual values for the convex program \eqref{eq:EG2}, and $p_e$ the corresponding prices. The following hold:
\begin{enumerate}[label=(\alph*),nosep]

    \item All goods are maximally allocated: For every $e$ in a part $P_i$ with $m_i \neq 0$, we have $p_e > 0$, so there is some $S \ni e$ such that $\alpha_S > 0$, and the corresponding primal constraint  is tight (i.e., $\sum_{e \in S} y_e = \rank(S)$).
    \item Each buyer $i$ is only buying elements (i.e. $y_e > 0$) of minimum price $q_i := \min_{e \in P_i} p_e$ (or equivalently, highest bang-per-buck). \label{prop:MI-2}
    \item Each buyer spends all of their money: $q_i \cdot \sum_{e \in P_i} y_{e} = \sum_{e \in P_i} p_e y_{e} = m_i$. \label{prop:MI-3}
    \item $\sum_{S \subseteq E} \alpha_{S} \cdot \rank(S) = \sum_{i\in B} m_i $. \label{prop:MI-4}
\end{enumerate}
\end{lem}
\begin{proof}
All parts follow directly from the above KKT optimality conditions.
\end{proof}

Jain and Vazirani~\cite{JV10} proposed a generalization of the Fisher
linear market called {\em Eisenberg-Gale markets}, which capture many
interesting markets, such as the resource allocation framework of
Kelly~\cite{Kelly97}. By definition, equilibria in these markets can
be computed using an Eisenberg-Gale type convex program. One specific
class of these markets are the so-called {\em submodular utility
  allocation (SUA)} markets, in which there are $n$ buyers (where
buyer $i$ has budget $m_i$) and each buyer has an associated utility
$u_i$. There are packing constraints on the utilities, which are
encoded via a {\em polymatroid} function
$\nu: 2^{[n]} \rightarrow \mathbb{R}_+$ (i.e., the function $\nu$ is
submodular, monotone, and $\nu (\emptyset) = 0$). The corresponding
SUA convex program is
\begin{gather}
  \max \big\{ \sum_i m_i \log u_i \mid \sum_{i \in S} u_i \leq \nu(S)
  ~~\forall S \subseteq [n], u \geq 0 \big\}.
\end{gather}
Note that in the MI market, we have an allocation constraint for each
subset of $E$, whereas in an SUA market there is an allocation
constraint only for each subset $\cup_{i \in A} P_i$, where
$A \subseteq [n]$. Nonetheless, an application of a continuous version
of Rado's theorem~\cite{McDiarmid75} can be used to prove the
following. (The proof is deferred to Appendix~\ref{sec:appendix}.)

\begin{restatable}{lem}{MIisSUA}
\label{lem:MIisSUA}
The MI market is a submodular utility allocation market.
\end{restatable}

\subsection{The Matroid Intersection Skeleton}
\label{sec:comb-decomp}

In this section, we give a combinatorial description of a specific set
of prices of elements with respect to the buyers' budgets. We refer to this decomposition as the \emph{matroid intersection (MI) skeleton} of the matroid intersection market. In
\Cref{lem:kkt-comb-decomp} we show these prices correspond to a very
specific optimal dual solution to the convex program \eqref{eq:EG2}.
Hence, although optimal dual solutions are not unique in general, this
allows us to focus on a unique set of duals, which we subsequently
show can be related to a notion of \emph{expansion}.

\begin{defn}
  Denote the budget for a subset of buyers $B' \subseteq B$ to be
  $m(B') := \sum_{i \in B'} m_i$. Also define the \emph{neighborhood}
  of these buyers as all elements in all of their parts:
  $N(B') := \bigcup_{i \in B'} P_i$. With this, define
  \emph{inverse expansion} of a set of buyers as
  \[ 
    \text{InvExp}_{\mathcal{M}}(B') := \frac{m(B')}{\rank_{\mathcal{M}}(N(B'))}
  \]
\end{defn}

We now give an algorithm that outputs a nested family of elements
$\varnothing = E_0 \subseteq E_1 \subseteq \ldots \subseteq E_L = E$,
which we call the matroid intersection skeleton, as well as prices for each of the items in
$E$. The algorithm, which is given as \Cref{alg:cap}, does the
following: in each step $\ell$, the algorithm finds the
(inclusion-wise maximal) set $B_{\ell+1}$ of buyers having maximum
inverse expansion $\rho$. The elements spanned by the neighborhood of
those buyers are assigned a price of $\rho$. These elements are then
contracted in $\mc{M}$, and the buyers $B_{\ell+1}$ are ``peeled
off'', with the process then repeating on the contracted matroid with
the remaining buyers. Defining $E_\ell$ to be the set of elements
contracted until the $\ell^{th}$ step gives us the desired nested
family
$\varnothing = E_0 \subseteq E_1 \subseteq \ldots \subseteq E_L =
E$. Moreover, the dual variables $\alpha_S$ can be chosen to be
positive only on the sets in this nested family. 

A similar nested family (of buyers, instead of elements) is introduced in an algorithm for finding equilibria in SUA markets \cite{JV07}; however, such a decomposition is not uniquely extendable to all items in an MI market. Our MI skeleton is
one such extension, which in particular, connects the market to a notion of expansion on the matroid intersection. 

Henceforth, we will use $\mc{M}_\ell$ to denote the matroid
contraction $\mc{M} / E_\ell$, and $\rank^{(\ell)}$ and
$\InvExp^{(\ell)}$ to denote the rank and the inverse expansion with
respect to this matroid $\mc{M}_\ell$. 
\begin{algorithm}
  \caption{The Matroid Intersection Skeleton}\label{alg:cap}
  Initialize $E_0 \gets \varnothing$, the contracted elements.\;
  Initialize $B^{\text{rem}}_0 \gets B$ the remaining buyers. \;
  \For{$\ell = 0, 1, \ldots$ until $B^{\text{rem}}_\ell$ is empty}{
    $\rho \gets \max_{B' \subseteq B^{\text{rem}}_\ell} \text{InvExp}^{(\ell)}(B')$.\;
    Find the (unique) largest set of buyers $B_{\ell + 1} \subseteq B^{\text{rem}}_\ell$ with inverse expansion $\text{InvExp}^{(\ell)}(B_{\ell + 1}) = \rho$.\;
    Consider the elements $S' := \Span_{\mc{M}_\ell}(N(B_{\ell+1}))$. \;
    Set the prices of each $e \in S'$ to be $\alg{p}_{\ell + 1} := \rho$. \;
    $E_{\ell + 1} \gets E_\ell \cup S'$.\;
    $B^{\text{rem}}_{\ell+1} \gets B^{\text{rem}}_\ell \setminus B_{\ell+1}$.
}
\end{algorithm}

\begin{restatable}{lem}{algWellDefined}
  \label{lem:algWellDefined}
  The MI skeleton is well-defined, and
  $\alg{p}_\ell > \alg{p}_{\ell + 1}$ for every $\ell$.
\end{restatable}
The proof of \Cref{lem:algWellDefined} uses the submodularity of the
matroid rank function, and uncrossing arguments to show the uniqueness
of the sets $B_\ell$, and the fact that the densities are strictly
increasing; we defer the formal argument to
\Cref{sec:appendix}. Instead we focus on showing that the prices
defined by \Cref{alg:cap} indeed give us optimal duals. The proof proceeds by
induction on the total number of rounds $L$. We begin with the case
$L = 1$ (where the ``densest'' set of buyers form the entire matroid).

\begin{claim}[Single Round]
  \label{claim:single-round-KKT}
  Consider an instance in which the inverse expansion of the entire
  set of buyers, $\rho := \InvExp_{\mc{M}}(B)$, is the maximum inverse
  expansion. Then any optimal primal-dual pair of solutions
  $\opt{y}, \opt{\alpha}$ for the convex program
  \eqref{eq:EG2} satisfies
  $\opt{p}_e := \sum_{S \ni e} \opt{\alpha}_S = \rho$. That
  is, all prices \emph{with respect to the optimal duals}
  $\opt{\alpha}$ are equal to $\rho$. One such optimal dual solution sets
  $\alg{\alpha}_E = \rho$ and $\alg{\alpha}_S = 0$ for all other sets.
 
\end{claim}

\begin{proof}
  Let $S^{\max} \sse E$ be the set of elements with price
  $\opt{p}_{\max} := \max_{e \in E} \opt{p}_e$. We claim that
  $S^{\max} = E$.

  As a first step, we show that constraint corresponding to $S^{\max}$
  in~\eqref{eq:EG2} is tight. This uses the observation that if the
  constraints for two sets are tight, then the constraints for their
  union and intersection are also tight due to submodularity of the
  matroid rank function. Now define:
  \[
    S_e := \bigcap_{\substack{S \ni e \\ \opt{\alpha}_S > 0}} S.
  \]
  This set $S_e$ is the smallest tight set containing $e \in S^{\max}$.
  Note that any $e' \in S_e$ has price at least that of $e$, and hence
  belongs to $S^{\max}$ by the maximality of $e$'s price. In turn this
  implies that $S^{\max} = \cup_{e \in S^{\max}} S_e$; writing $S^{\max}$
  as a union of tight sets shows its tightness.

  Next, define $B^{\max}$ as the set of buyers buying from $S^{\max}$;
  we claim that $N(B^{\max})$ is also tight. Indeed,
  $N(B^{\max}) \sse S^{\max}$, since the buyers in $B^{\max}$ must buy
  at the least price in their neighborhood. All other items in
  $S^{\max} \setminus N(B^{\max})$ must have $y_e = 0$ since they are not
  sold.  This means
  $y(N(B^{\max})) = y(S^{\max}) = r(S^{\max}) \geq r(N(B^{\max}))$; combining
  with the constraint in \eqref{eq:EG2} means the last inequality is
  an equality.

  Now, observe that for any subset $K$ of buyers, we have
  \begin{equation}\label{eq:inv-exp<max-price}
    m(K) = \sum_{i \in K} \opt{q}_i \sum_{e \in P_i} y_e \leq
    \opt{p}_{\max} \sum_{i \in K} \sum_{e \in P_i} y_e \leq
    \opt{p}_{\max} \rank(N(K)), 
  \end{equation}
  where the equality uses
  \Cref{thm:MI-market-conditions}\ref{prop:MI-3}. In particular, this
  implies that $\InvExp_{\mc{M}}(K) \leq \opt{p}_{\max}$ for every
  $K \subseteq B$, and so $\rho = \max_K \InvExp_{\mc{M}}(K) \leq \opt{p}_{\max}$.

  Now consider setting $K = B^{\max}$: the first inequality in
  \cref{eq:inv-exp<max-price} holds at equality because each buyer in
  $B^{\max}$ buys at price $\opt{p}_{\max}$, while the second holds at
  equality by the tightness of $N(B^{\max})$. So
  $\InvExp_{\mc{M}}(B^{\max}) = \opt{p}_{\max}$, which shows that 
  $\rho \geq \opt{p}_{\max}$, and in fact equality holds.

  Finally, for $K = B$,
  \cref{eq:inv-exp<max-price} implies
  \[
    m(B) \leq \opt{p}_{\max}\rank(N(B)) = \rho \rank(N(B)) = m(B)
  \]
  by the assumption that $\InvExp_{\mc{M}}(B) = \rho$. Therefore,
  setting $K=B$ gives equality throughout the sums in
  \cref{eq:inv-exp<max-price}. That is, $\opt{q}_i = \opt{p}_{\max}$
  for all $i \in B$, which means that each buyer $i$ buys at price
  $\opt{p}_{\max}$. In particular, $B^{\max} = B$ and
  $S^{\max} = N(B) = E$.

  The fact that $\alg{\alpha}$ is an optimal dual follows from
  checking the KKT conditions with respect to $y^*$: \Cref{eq:KKT1} is
  satisfied trivially, and \Cref{eq:KKT2} is satisfied since we have
  shown above that $E$ is a tight set.
\end{proof}

\Cref{claim:single-round-KKT} shows that the prices we computed (for the
case of $L=1$) indeed correspond to optimal duals. Next, we extend the
proof to the case of general $L$, via an inductive
argument.

\begin{lem}[Prices are Optimal]
  \label{lem:kkt-comb-decomp}
  Consider the prices $\alg{p}_1, \ldots, \alg{p}_L$ from~\Cref{alg:cap}, and define
  \begin{align*}
    \alg{\alpha}_{E_L} &:= \alg{p}_L\\
    \alg{\alpha}_{E_{\ell}} &:= \alg{p}_\ell - \alg{p}_{\ell + 1} \qquad \ell = 1, \ldots, L-1
  \end{align*}
  and $\alg{\alpha}_S = 0$ for all other $S \subseteq E$; note that
  all $\alg{\alpha}_S \geq 0$ by \Cref{lem:algWellDefined}. Then the
  resulting $\alg{\alpha}$ forms an optimal dual solution to
  \eqref{eq:EG2}; in particular, it satisfies the KKT optimality
  conditions.
\end{lem}

\begin{proof}

  We construct an optimal primal solution $\alg{y}$, such that
  $(\alg{y}, \alg{\alpha})$ form a primal/dual pair satisfying the
  KKT conditions for the convex program in \eqref{eq:EG2}.

  We proceed via induction on $L$, the number of steps in \Cref{alg:cap}. \Cref{claim:single-round-KKT} precisely proves the base case of $L = 1$. For our inductive hypothesis, assume the lemma holds true for \Cref{alg:cap} having up to $k$ iterations. We consider an instance that requires $k + 1$ iterations of the \Cref{alg:cap}. 
  Consider the two markets
\begin{itemize}[nosep]
    \item market $\Market^{(1)}$ on buyers $B_1$, where all budgets $m_i$ are equal to those from market $\Market$, and with the matroid $\mc{M}$ restricted to the ground set $E^{(1)}:= N(B_{1})$. Let $y^{(1)}$ be an optimal primal solutions. By \Cref{claim:single-round-KKT}, the vector $\alg{\alpha}^{(1)}$ where $\alg{\alpha}^{(1)}_{N(B_1)} = p_1$ (and all other sets have $\alg{\alpha}$ value 0) is an optimal dual solution to market $\Market^{(1)}$. 

    \item market $\Market^{(2)}$ on buyers $B \setminus B_1$, where all
      budgets $m_i$ are equal to those from market $\Market$, and with the
      matroid $\mc{M} / \Span(E^{(1)})$ (contracting elements
      $\Span(E^{(1)})$). We apply the inductive hypothesis on
      market $\Market^{(2)}$: Let $y^{(2)}$ be an optimal primal
      solution and let $\alg{\alpha}^{(2)}$ be the
      (optimal) dual solution arising from \Cref{alg:cap}. 
\end{itemize}
We may consider a gluing of $y^{(1)}$ and $y^{(2)}$ to create a primal solution $\alg{y}$
\begin{align*}
\alg{y}_e :=
    \begin{cases}
      y^{(1)}_e  & \text{if $e \in E^{(1)}$} \\
      0 & \text{if $e \in \Span(E^{(1)}) \setminus E^{(1)}$} \\
      y^{(2)}_e & \text{otherwise}
    \end{cases}
\end{align*}
The proof that $\alg{y}$ and $\alg{\alpha}$ satisfy the KKT conditions for market $\Market$ is an application of the gluing lemma, \Cref{lem:gluing}. \end{proof}

\begin{defn}[Canonical Duals/Prices]
  The duals and prices for \eqref{eq:EG2} generated by \Cref{alg:cap}
  are called \textit{canonical duals and prices}.
\end{defn}
\Cref{thm:comb-decomp} shows that canonical duals are optimal duals
for the convex program~\eqref{eq:EG2}. We can now extend
\Cref{thm:MI-market-conditions} to record two more properties that are
specific to canonical prices.

\begin{thm}[Two More Properties]
  \label{thm:comb-decomp}
  Let $\alpha_{S}$ be the canonical dual values for \eqref{eq:EG2} with corresponding prices $p_e$, and let $y$ be any optimal primal solution for \eqref{eq:EG2}. The following hold:
  \begin{enumerate}[label=(\alph*),nosep]
  \addtocounter{enumi}{4}
  \item\label{prop:MI-5} Let $F \subseteq E$ be a maximum-price basis of
    $\mc{M}$. Then $\sum_{e \in F} p_e =
    m(B)$. 
  \item\label{prop:MI-6} For $B' \subseteq B$ and $q_{\max} := \max_{i \in
      B'} q_i$, consider the set of elements $S := \{e \in
    N(B'): y_e > 0 \}$, and $E_>$ the elements with prices higher
    than $q_{\max}$. Then $\rank_{\mc{M} / E_>}(S) \geq
    \frac{m(B')}{q_{\max}}$.
  \end{enumerate}
\end{thm}
\begin{proof}
  We prove:
  \begin{enumerate}
  \item[(e)] Let $F$ be a maximum-price basis (with respect to
    $\mc{M}$) of elements. We first claim that for any
    $S \subseteq E$, if $\alpha_{S} > 0$ (in particular,
    $S = E_\ell$ for some $\ell$), then $F$ contains $\rank(S)$
    elements from $S$. Observe that for each $E_\ell$ from the
    MI skeleton, the greedy algorithm to compute $F$
    always considers all elements of $E_\ell$ before considering
    elements of $E \setminus E_{\ell}$ (those of lower
    price). Therefore, $\abs{F \cap E_\ell} = \rank(E_\ell)$.
    
    Finally, \Cref{thm:MI-market-conditions}\ref{prop:MI-4} and the
    above claim together imply the desired:
    \[
        m(B) = \sum_{S : \alpha_{S} > 0} \alpha_{S} \cdot
        \rank(S) = \sum_{\ell} \sum_{e \in F \cap E_\ell}
        \alpha_{E_\ell} = \sum_{e \in F} \sum_{S \ni e}
        \alpha_S = \sum_{e \in F} p_e. 
    \]
    
  \item[(f)] Denote $q_{\max} := \max_{i \in B'}
    q_i$. First observe that since $y_e > 0$ for $e \in S$
    (lying in part $P_i$ for some $i \in B'$),
    \Cref{thm:MI-market-conditions}\ref{prop:MI-2} implies that
    $p_e = q_i \leq q_{\max}$. Now consider \Cref{alg:cap}, in particular, the step where all
    elements with prices strictly greater than $q_{\max}$ have been
    contracted. This is exactly the matroid $\mc{M} / E_>$. At this
    point, the maximum inverse expansion is at most $q_{\max}$
    (by \Cref{lem:algWellDefined}). Hence, for any subset of buyers
    $B'$ at this stage, we have
    \[ \frac{m(B')}{\rank_{\mc{M}/E_>}(M(B'))} \leq q_{\max}. \]
    This is precisely the desired result. \qedhere
  \end{enumerate}
\end{proof}

Part \ref{prop:MI-5} is used later in the proof of the main theorem to bound the
total price, while part \ref{prop:MI-6} gives a relation between the prices in an
arbitrary set of buyers and its expansion (in some contraction of
$\mc{M}$). 

\subsection{Monotonicity of Prices}
\label{subsec:monotonicity}

We can now show that the prices of items only increase as new buyers
arrive. Showing this for the minimum prices that buyers observe
follows from a property called \emph{competition monotonicity} that
holds for SUA markets (and hence for our MI market). However, for the proof of \Cref{thm:main-result} we will need
to use the fact that the canonical prices are monotone for {\em all} items (and not just
for the price that each buyer buys at). Note that since the item prices are not unique, monotonicity of
prices may not hold for all items.

For the purposes of this proof, we view the arrival of buyers as a change in their budgets. That is, the budget of all buyers that have arrived is 1 and the budget of all buyers that have not yet arrived is 0.
Thus, the arrival of $i^{\text{inc}}$ is represented by an increase in its budget from $m_{i^{\text{inc}}}= 0$  to $m_{i^{\text{inc}}}= 1$. Throughout this section, we assume that the MI skeleton immediately before $i^{\text{inc}}$ arrives is given by the sequence of buyers $B_1, \ldots, B_L$, and the nested family of elements $E_1 \subseteq \ldots \subseteq E_L = E$, where price $p_\ell$ is assigned to the elements of $E_\ell\setminus E_{\ell-1}$, and each buyer $i \in B_\ell$ buys at price $q_{i}^{\text{old}} = p_\ell$ (i.e., $p_\ell$ is the minimum price element in $P_i$). We will use the notation $B_{\le k}:= \cup_{j=1}^k B_j$.

\begin{thm}[Monotonicity of Prices]
  \label{thm:monotonicity}
  Let $i^{\text{inc}}$ be an incoming buyer and let
  $p^{\text{old}}, p^{\text{new}}$ be the canonical price vectors for
  the instances before and after its arrival. Then, for every item
  $e \in E$,
\begin{equation}
\label{eq:monotone}	
p^{\text{new}}_e \ge p^{\text{old}}_e 
\end{equation}
Moreover, let $q:= \min_{e \in P_{i^{\text{inc}}}}\{p_e^{\text{old}}\}$ be the minimum price that buyer $i^{\text{inc}}$ observes before arriving. If $p^{\text{old}}_e < q$, then, 
\begin{equation}
\label{eq:stability}	
p^{\text{new}}_e =p^{\text{old}}_e
\end{equation}
\end{thm}

We will prove \Cref{thm:monotonicity} in two steps: first, we show monotonicity of the
prices buyers are buying at\footnote{For a buyer $i$ with a budget of 0, this price is simply the price of the cheapest element of $P_i$}. Second, we use the structure of the
MI skeleton to extend this monotonicity to all
elements.

\begin{lem}
\label{lem:competition-monotonicity}
For any buyer $i$, let $q^{\text{old}}_{i} = \min_{e \in P_{i}} p_e$ be the price it buys at immediately before the arrival of $i^{\text{inc}}$, and similarly $q^{\text{new}}_{i}$ the price it buys at immediately after. Then $q^{\text{new}}_i \geq q^{\text{old}}_i$. 
\end{lem}

\begin{proof}
Recall that by \Cref{lem:MIisSUA}, the matroid intersection market is an SUA market. Jain and Vazirani~\cite{JV07} proved the following monotonicity result about SUA markets: 
\begin{fact}[Competition monotonicity]
At the addition of a new buyer $i^{\text{inc}}$, no other buyer's utility $u_i$, defined in MI markets as $u_i := \sum_{e \in P_i} y_e$, increases. 
\end{fact}
An immediate corollary of this fact and \Cref{thm:MI-market-conditions}, part \ref{prop:MI-3} is that for all buyers $i$ that arrive before $i^{\text{inc}}$, we have $q^{\text{new}}_i \geq q^{\text{old}}_i$.

It remains to show that the minimum price of elements in $P_{i^{\text{inc}}}$ also does not decrease. To see this, suppose that (immediately before its arrival) $i^{\text{inc}}$ is in $B_k$.  Then, the MI skeleton implies that
\begin{equation}
\label{invExpIneq}
    \InvExp^{(k)} (B_{k}) \ge \InvExp^{(k)} (B_{k}\setminus  i^{\text{inc}} ),
\end{equation}
but since $m(B_{k}) =m(B_{k}\setminus  i^{\text{inc}})$, we must have  $\rank^{(k)} (N(B_{k})) = \rank^{(k)} (N(B_{k}  \setminus i^{\text{inc}}))$, or $ i^{\text{inc}}$ would not be in $B_{k}$.
Therefore, $P_{i^{\text{inc}}} \subseteq \Span(N(B_{\le k} \setminus i^{\text{inc}}))$. Therefore, since the prices paid by buyers in $B_{\le k} \setminus i^{\text{inc}}$ do not decrease after the arrival of $i^{\text{inc}}$, once all of these buyers are peeled off, $i^{\text{inc}}$ must be as well, at a price no smaller than its price was before its arrival.
\end{proof}

Now we can prove monotonicity of the prices of all elements.

\begin{proof}[Proof of \Cref{thm:monotonicity}] 
We first prove \Cref{eq:monotone} and specifically that $p_e^{\text{new}} \geq p_e^{\text{old}}$ for each $e \in E_\ell \setminus E_{\ell-1}$ by induction on $\ell$. There are two cases: since $e$ was contracted at step $\ell$ of \Cref{alg:cap}, we must either have that $e \in N(B_\ell)$, or $e \in \Span_{\mc{M}_{\ell-1}}(N(B_\ell)) \setminus N(B_\ell)$.
\begin{enumerate}
    \item If $e \in N(B_\ell)$, then $e \in N(i)$ for some $i \in B_{\ell}$, and $p_e^{\text{old}} = p_\ell = q_{i}^{\text{old}}$. But then by \Cref{lem:competition-monotonicity}, 
    \[
        p^{\text{new}}_e \geq q^{\text{new}}_{i} \geq q_{i}^{\text{old}} = p_e^{\text{old}}.
    \]

    \item If $e \in \Span_{\mc{M}_{\ell-1}}(N(B_\ell))  \setminus N(B_\ell)$, then we have that $e \in \Span_{\mc{M}}(E_{\ell-1} \cup N(B_\ell))$, by the definition of matroid contraction. By induction and case 1 above, we know that all elements $e'$ in $E_{\ell-1} \cup N(B_\ell)$ have new price at least $p_{e'}^{\text{new}} \geq p_{e'}^{\text{old}} \geq p_\ell = p_e^{\text{old}}$. Moreover, by design of \Cref{alg:cap}, we must have that after the arrival,
    \[
        p^{\text{new}}_e \geq \min_{e' \in E_{\ell-1} \cup N(B_\ell)} p^{\text{new}}_{e'} \geq p_e^{\text{old}}
    \]
    where the first inequality follows because once all elements in $E_{\ell-1} \cup N(B_{\ell})$ have been assigned a price, then so have the elements in their span, including $e$. 
\end{enumerate}

It remains to prove \Cref{eq:stability}, that $p^{\text{new}}_e =p^{\text{old}}_e$ if $ p^{\text{old}}_{e} < q $. Recall that when running \Cref{alg:cap}, immediately before $i^{\text{inc}}$'s arrival,  $i^{\text{inc}}$ is peeled off at step $k$ with corresponding price $p_k$.  We claim that this implies that in the MI skeleton with $ m(i^{\text{inc}}) = 1$, all buyers of $B^{(1)} := B_{\le k}$ are peeled off before any buyer $i \in B^{(2)} := B \setminus B^{(1)}$, and hence the prices of all elements with $p_e^{\text{old}} < p_k$   are unchanged. 

To prove this, we consider two submarkets of the MI market after the arrival of $i^{\text{inc}}$ : (1) market instance  $\Market^{(1)}$ with  buyers $B^{(1)}$, elements $E^{(1)} := N(B^{(1)})$, with matroid $\mc{M}^{(1)} := \mc{M}|_{E^{(1)}}$ and partition matroid $\mc{P}|_{E^{(1)}}$, and (2) market instance $\Market^{(2)}$ with buyers $B^{(2)} := B \setminus B^{(1)}$, $E^{(2)} := E \setminus \Span(E^{(1)})$,  matroid $\mc{M}^{(2)} := \mc{M} / \Span(E^{(1)})$, and partition matroid $\mc{P}|_{E^{(2)}}$. Budgets in both markets are precisely those immediately after the arrival of $i^{inc}$. Note that deriving an MI skeleton from each of these markets separately yields optimal primal and dual solutions $y^{(1)}, \alpha^{(1)}$ (for market $\Market^{(1)}$) and $y^{(2)}, \alpha^{(2)}$ (for market $\Market^{(2)}$) where all prices in market $\Market^{(1)}$ are at least $p_{k}$ and all prices buyers in market $\Market^{(2)}$ buy at are exactly as they were from the MI skeleton on that market immediately before $i^{inc}$ arrived. Note also that the elements in  $\Span(E^{(1)})\setminus E^{(1)}$ are not included in either market.
To complete the proof then, it suffices to show that these two solutions can be ``glued together'' to obtain an optimal primal and dual solution to the full market after the arrival of $i^{\text{inc}}$, where the prices of all elements in $E^{(2)}$ are given by  the optimal solution to market $\Market^{(2)}$ and thus are unchanged relative to the prices before $i^{inc}$ arrived. This is shown in \Cref{lem:gluing}.
\end{proof}

\subsection{The Expansion Lemma}
\label{sec:expansion-matroids}

We now prove our expansion lemma, which relates the length of
augmenting paths in the exchange graph to prices of elements. 

\begin{lem}\label{lem:laminar-expansion}
  Let $\Ind$ be an arbitrary independent set in the intersection
  $\mc{P} \cap \mc{M}$, and let $p \in \RR^{|E|}_{\geq 0}$ be
  the canonical prices. Then, for any element $e^* \in E$ with current price
  $p_{e^*} \in [0, 1)$ such that $\Ind \cup e^* \in \mc{P}$, there is an augmenting path from $e^*$ of
  length at most $O\big(\frac{\ln n}{1 - p_{e^*}}\big)$.
\end{lem}

Just as in the case of matchings, we define a sequence of sets of
buyers $\{\Left_k\}_k$ and of items $\{\Right_k\}_k$ inductively,
based on the independent set $\Ind$ and the prices $p$. Firstly,
for an element $e \in E$, let $\circuit(e, \Ind) \sse \Ind$ be the elements
from $\Ind$ in the circuit formed by adding $e$ to $\Ind$. (If
$\Ind \cup \{e\} \in \cM$ then $\circuit(e, \Ind) = \emptyset$.) Now, for any
$S \sse E$, we define
\[ \circuit(S, \Ind) := \bigcup_{e \in S} \circuit(e, \Ind). \] Secondly, for a
set $S \sse E$, let
$S_{\leq \tau} := \{ e \in S \mid p_e \leq \tau\}$ denote the
set of elements in $S$ which have price less than or equal to some
threshold $\tau$. With this, we may define our sequences of
sets: 
\begin{align*}
  \Right_1 &:= \{ e^* \} \\
  \Left_k &:= \{ i \in B \mid P_i \cap \circuit(\Right_k, \Ind)_{\leq p_{e^*}} \neq \emptyset \} \\
  \Right_{k + 1} &:= \{ e \in N(\Left_k) \mid y_e > 0 \}.
\end{align*}
In order to prove~\Cref{lem:laminar-expansion}, we need the following
expansion claim.

\begin{claim}
  If $\Right_1, \ldots, \Right_{k}$ do not contain a free element with
  respect to $\cM$, then
  \begin{enumerate}[nosep,label=(\alph*)]
  \item $p_e \leq p_{e^*}$ for all $e \in \Right_{k+1}$, and
  \item $\rank_\cM(\Right_{k+1}) \geq (\nicefrac{1}{p_{e^*}})^{k}$.
  \end{enumerate}
\end{claim}

\begin{proof}
  We proceed by induction on $k$. The base case is
  $k=0$, in which case both properties are true for the singleton set
  $\Right_1 = \{e^*\}$. Now, to inductively prove property~(a): consider
  an element $e \in \Right_{k+1}$, and say it belongs to part
  $P_i$. Since $y_e > 0$,
  \Cref{thm:MI-market-conditions}\ref{prop:MI-2} implies that
  $p_e = q_i$. Moreover, $i \in \Left_k$ by definition, so part $P_i$ contains an
  element from $\circuit(\Right_k, \Ind)_{\leq p_{e^*}}$, which in
  particular implies that it contains an element with price at most
  $p_{e^*}$ and so $q_i \leq p_{e^*}$. Putting these two facts
  together shows $p_e \leq p_{e^*}$, and hence part~(a).  To prove
  part~(b): for brevity, define the sets
  \[ K := \circuit(\Right_k, \Ind) \qquad \qquad \text{and} \qquad
    \qquad K_\leq := \circuit(\Right_k, \Ind)_{\leq p_{e^*}}. \]
  Since $K_\leq \sse K \sse \Ind$, and $\Ind$ is independent in the
  matroid $\cI$, we have $\rank_\cI(K_\leq) = |K_\leq|$. This means
  that each $i \in \Left_k$ has exactly one element in
  $P_i \cap K_\leq$, and therefore $|\Left_k| = |K_\leq|$.  Let $E_>$
  be all the elements of price greater than $p_{e^*}$, and consider
  the contraction $\cM_\leq := \mc{M} / E_>$. We claim that
  \begin{gather}
    \rank_{(\cM_\leq)}(\Right_k) \leq \abs{K_{\leq}} =
    \abs{\Left_k}. \label{eq:claim-mono}
  \end{gather}
  Before we prove the claim in~\eqref{eq:claim-mono}, let us use it to
  complete the inductive proof for part~(b). Indeed,
  \Cref{thm:comb-decomp}\ref{prop:MI-6}, along with the fact that no client of
  $\Left_k$ buys at price greater than $p_{e^*}$, tells us that
  \[
    \rank_{(\cM_\leq)}(\Right_{k + 1}) \geq
    \frac{\abs{\Left_k}}{p_{e^*}} .
  \]
  Using~\eqref{eq:claim-mono} now proves the second part of the
  inductive claim:
  \[ \abs{\Right_{k + 1}} \geq \rank_{(\cM_\leq)}(\Right_{k + 1}) \geq
    \frac{\rank_{(\cM_\leq)}(\Right_k)}{p_{e^*}} \geq
    \left(\frac{1}{p_{e^*}}\right)^k. \]

  Finally, it remains to prove the inequality in~\eqref{eq:claim-mono}.
  For the sake of contradiction, suppose
  \[ |\Right_k| \geq \rank_{(\cM_\leq)}(\Right_k) > \abs{K_{\leq}} \geq \rank_{(\cM_\leq)}(K_{\leq}) 
    . \] By the matroid exchange property,
  there is some $e \in \Right_k$ for which $K_{\leq} \cup \{e\}$ is
  independent. But by the definition of matroid contraction, this
  implies that in our un-contracted matroid $\mc{M}$ we have
  \[\rank_{\mc{M}}(K_{\leq} \cup \{e\} \cup E_>) >
    \rank_{\mc{M}}(K_{\leq} \cup E_>).\] Since
  $K \subseteq K_{\leq} \cup E_>$, submodularity gives
  $\rank_{\mc{M}}(K \cup e) > \rank_{\mc{M}}(K)$. But this contradicts
  the fact that $K$ is independent, while $K \cup e$ contains the
  circuit $\circuit(e, \Ind)$ by construction.
\end{proof}

\begin{proof}[Proof of \Cref{lem:laminar-expansion}]
  Indeed, if $\Right_k$ does contain a free element with respect to
  $\cM$, then there is a path in the exchange graph from $e^*$ to this
  free element of length at most $2k$. Indeed, each element
  $e \in \Right_k$ is contained in $P_i$ for some buyer
  $i \in \Left_{k-1}$, and $i$ has a neighbor
  $e' \in \circuit(\Right_{k-1}, \Ind)_{\leq p_{e^*}} \subseteq
  \Ind$. This means that there is some $e'' \in \Right_{k-1}$ for
  which $\Ind -e' + e'' \in \mc{M}$, and $\Ind - e' + e \in \mc{P}$.
  Inductively, $e''$ is reachable from $e^*$ by an alternating path of
  length $2(k-1)$, so therefore $e$ is reachable by an alternating
  path of length $2k$. Such an alternating path to a free element
  implies that the \emph{shortest} AP to a free element is of length
  no more than $2k$.

  On the other hand, if $\Right_k$ contains no free element, then we
  have
  $(\nicefrac{1}{p_{e^*}})^{k-1} \leq \rank_{\cM}(\Right_k) \leq
  \abs{\Left_k} \leq n$. Taking logs, we get that $k$ is at most
  $\smash{O(\max\{1, \frac{\ln n}{\ln(1/p_{e^*})}\})} = O(\frac{\ln
    n}{1-p_{e^*}})$ when the process ends with an augmenting path of
  length at most $2k$, hence proving~\Cref{lem:laminar-expansion}.
\end{proof}

\subsection{Bounding the Total Augmentation Cost}

\begin{proof}[Proof of \Cref{thm:main-result}]
When the $i$th buyer arrives, say the minimum price neighbor before its arrival (i.e. while it has 0 budget) has price $q_{\min}$. Denote the maximum-price basis before arrival $i$ as $F_{i-1}$, having elements of prices $f_1^{(i-1)} \geq \ldots \geq f_{R}^{(i-1)}$ (where $R$ denotes the rank of $\mc{M}$), and similarly for $F_i$ after arrival $i$. $\Delta f_z^{(i)} := f_z^{(i)} - f_z^{(i-1)} \geq 0$, by \Cref{thm:monotonicity} (monotonicity of prices). Then, using the Expansion lemma (\Cref{lem:laminar-expansion}), the length $\ell_i$ of the shortest augmenting path upon arrival $i$ is at most
\[
    \ell_i \leq O\left(\frac{\ln n}{1-q_{\min}} \right)= O\left(\frac{\ln n}{1-q_{\min}} \right) \cdot \sum_{z = 1}^{R} \Delta f_z^{(i)} \leq O(\ln n) \cdot \sum_{z = 1}^{R} \frac{\Delta f_z^{(i)}}{1-f_z^{(i-1)}}.
\]
where the equality follows from \Cref{thm:comb-decomp}\ref{prop:MI-5} (since the total budget increases by 1 at the arrival of $i$), and the final inequality again from \Cref{thm:monotonicity}, since either $\Delta f_z^{(i)} = 0$, or $q_{\min} \leq f_z^{(i-1)}$. 

Denote $\widehat{f}_z := \max\{f_z^{(i)} : f_z^{(i)} < 1\}$. Notice that by \Cref{thm:comb-decomp}\ref{prop:MI-5} at most $n$ elements in $F_n$ can have final price $f_z^{(n)} = 1$. The total cost of all augmenting paths is at most:
\begin{align*}
    \sum_{i=1}^{n}\ell_i &\leq O(\ln n) \sum_{z = 1}^{R} \sum_{i = 1}^n \frac{\Delta f_z^{(i)}}{1-f_z^{(i-1)}} \\
    &\leq O(\ln n) \sum_{z = 1}^{R} \left(\mathds{1}_{(z \leq n)} + \int_{p=0}^{\widehat{f}_z}\frac{dp}{1-p}\right) \\
    &= O(n\ln n) + O(\ln n) \sum_{z = 1}^{R} \ln\left( \frac{1}{1-\widehat{f}_z} \right)
\end{align*}

It remains to upper bound the right hand term by $O(n \ln^2 n)$. Here we give a more careful argument of this fact than in \S\ref{sec:matchings}. We can compute the maximum value the sum could possibly attain, given certain constraints on the prices $\widehat{f}_z$. Observe that at all times, prices satisfy $p_e\leq \frac{n}{n+1}$ or $p_e=1$ for every element~$e$. Indeed, \Cref{alg:cap} assigns $e$ a price $p = \frac{m(B')}{\rank^{(\ell)}(N(B'))}$ for some $\ell$ and some set $B'$ of buyers. If $p < 1$, then $\rank^{(\ell)}(N(B')) \geq m(B') + 1$, so $p \leq \frac{m(B')}{m(B') + 1} \leq \frac{n}{n+1}$. Also from part~\ref{prop:MI-5} of \Cref{thm:comb-decomp}, $\sum_{z} \widehat{f}_z \leq \sum_{z} f_z^{(n)} = n$. Now, we have
\[
    \left\{\begin{array}{ll}
        \max_p &\sum_{z = 1}^R \ln\left(\frac{1}{1-p_z}\right) \\
        \text{s.t.} &\sum_{z = 1}^R p_z \leq n\\
        &p_z \leq \frac{n}{n+1}
    \end{array}\right\} = \left\{\begin{array}{ll}
        \max_p &\sum_{z = 1}^R \ln\left(\frac{1}{1-p_z}\right) \\
        \text{s.t.} &\sum_{z = 1}^R p_z = n\\
        &p_z \leq \frac{n}{n+1}
    \end{array}\right\} = (n+1)\ln(n+1).
\]
The first equality follows from monotonicity of the objective. The second comes from the following mass-shifting argument: given two elements with prices $p, q$ and $p \geq q$, let their sum be $p + q = D$. Then 
\[
    \frac{d}{dp}\left[\log\left(\frac{1}{1-p}\right) + \log\left(\frac{1}{1-(D-p)}\right)\right] = \frac{1}{1-p} + \frac{1}{1-(D-p)} \geq 0.
\]
In particular, we can increase the sum by moving price-mass from $p_i$ to $p_j$ for any $p_i < p_j$. Therefore the maximizing $p$ has $p_z = \frac{n}{n+1}$ for $n+1$ elements, and $p_z = 0$ for the rest. Since our prices $\{\widehat{f}_z\}$ satisfy the constraints in the maximization problem above, we have shown $\sum_{i=1}^{n}\ell_i \leq n \ln n + (n+1)\ln(n+1)\ln n = O(n\ln^2n)$. 
\end{proof}

\section{Closing Remarks}
\label{sec:closing}

We gave the matroid intersection maintenance problem, where we
maintain a common base in the intersection of a partition matroid
$\cP$ and another arbitrary matroid $\cM$, where the parts of $\cP$
appear online, such that the total number of changes performed is
$O(n \log^2 n)$. This extends the previous result for the special case
of the intersection of two partition matroids (i.e., bipartite
matching). Our results were based on viewing the problem from the
perspective of market equilibria, and using market-clearing prices to
bound the lengths of augmenting paths. Several open problems remain:
the most natural one is whether we can improve the bound to
$O(n \log n)$, or give a better lower bound, even for the bipartite
matchings case. Can a better bound be given for the fractional variant
of the problem? Can the matroid constraints be generalized to broader
sets of packing constraints (again in the fractional case)? 

Finally, what can we say about intersections of two arbitrary
matroids? This problem is hopeless in full generality because it 
captures the edge-arrival model in bipartite matchings, which has 
a simple $\Omega(n^2)$ lower bound. Are there other interesting 
special cases which avoid these lower bounds?

\subsection*{Acknowledgments}

Part of this work was conducted while A.\ Gupta and S.\ Sarkar were visiting the Simons Institute for the Theory of Computing; they thank the institute for its generous hospitality. 

\appendix
\appendix

\section{Appendix}
\label{sec:appendix}
\MIisSUA*
\begin{proof}
In a sub-modular utility allocation (SUA) market, the utility of buyers are defined by sub-modular packing constraints. The convex program corresponding to this market is
\begin{equation}\label{eq:convex-prog-SUA}
\begin{aligned}
    \max\quad &\sum_{i \in B} m(i) \cdot \log u_i \\
     \quad &\sum_{i \in B'} u_i \leq \nu(B') \quad \forall B' \subseteq B && \quad (\alpha_{B'}) \\
    &u_{i} \geq 0
\end{aligned}
\end{equation}
where $\nu$ is a sub-modular, monotone function. We consider the following SUA market. Define for a subset of buyers $B' \subseteq B$ the sub-modular monotone function
\[ \nu(B') := \text{rank}(N(B')). \]
We will show that this SUA market and the matroid intersection market are equivalent, i.e., an optimal set of $y_e's$ can be transformed into an optimal solution in the SUA market, and vice versa. To see that a solution to the matroid intersection market can be turned into a solution to the SUA market with the same objective value, consider the following assignment of $u_i$'s: 
\[ u_i := \sum_{e \in P_i} y_e.\]
Clearly, the $u_i$'s are feasible, and the objective values are the
same. 

For the converse direction, we use the following theorem of McDiarmid, an extension of Rado's theorem
\begin{thm}[Proposition 2C of \cite{McDiarmid75}]
Let $G = (X, Y, E)$ be a bipartite graph, and $\mathbb{P}$ a polymatroid on ground set $Y$ associated to polymatroid function $\rho$. Let $u \in \RR^+_{X}$, then $u$ is linked onto some vector $y \in \mathbb{P}$ if and only if 
\begin{equation}\label{eq:rado}
    u(B') \leq \rho(N_G(B')) \qquad \forall B' \subseteq X. 
\end{equation}
\end{thm}
In McDiarmid's language, any solution $u$ in the SUA market (and therefore satisfying~(\ref{eq:rado})) is ``linked to'' some $y$ in the polymatroid defined by polymatroid function $\rho(A) = \rank_{\mc{M}}(A)$ (so $\mathbb{P}$ is the matroid polytope of $\mc{M}$). This precisely gives us a solution of $y_e$'s that lie in the matroid polytope of $\mc{M}$ such that $u_i = \sum_{e \in P_i} y_e$.
\end{proof}

\algWellDefined*

\begin{proof}
In particular, we have the following:
\begin{enumerate}[label=(\alph*),nosep]
    \item Each $B_{\ell+1}$ is uniquely defined.

    \item At each step $\ell$, every remaining buyer $i \in B^{\text{rem}}_\ell$ has a nonempty neighborhood. That is, $P_i \setminus E_\ell \neq \varnothing$. 
    
    \item The prices assigned in step $\ell$ are greater than those assigned in later steps. That is, $\alg{p}_\ell > \alg{p}_{\ell + 1}$. 
\end{enumerate}

First, we state a fact that will be useful:
\begin{fact}
\label{fact:fractions}
For non-negative $a, b, c, d, p$, if $\frac{a}{b} \leq p$ and $\frac{c}{d} \leq p$, then $\frac{a + c}{b + d} \leq p$. Moreover, if either of the first two inequalities are strict, then so is the third.
\end{fact}
\begin{enumerate}[label=(\alph*)]
    \item Let $B', B''$ be two sets of buyers at iteration $\ell$ with maximum inverse expansion $\rho$. We claim that $B' \cup B''$ has the same inverse expansion. Observe that
    \begin{align*}
        m(B' \cup B'') &\leq \rank^{(\ell)}(N(B') \cup N(B'')) \cdot \rho \\
        &\leq \left(\rank^{(\ell)}(N(B')) +\rank^{(\ell)}(N(B'')) - \rank^{(\ell)}(N(B') \cap N(B'')) \right) \cdot \rho\\
        &\leq m(B') + m(B'') -m(B' \cap B'')\\
        &= m(B' \cup B''),
    \end{align*}
    where the first inequality follows from maximality of $\rho$, the second by submodularity, and the third by our assumption on the inverse expansion of $B'$ and $B''$. Therefore equality holds throughout, and $\InvExp^{(\ell)}(B' \cup B'') = \rho$. So the union of all sets of maximum inverse expansion gives the unique largest set of maximum inverse expansion. 
    
    \item Consider for contradiction the first iteration $\ell+1$ at which some remaining buyer $i \in B^{\text{rem}}_{\ell+1}$ has empty neighborhood. That means $i \not \in B_{\ell+1}$, while its neighbors $P_i \setminus E_{\ell}$ are in $\Span_{\mc{M}_{\ell}}(N(B_{\ell + 1}))$. But then we could add $i$ to $B_{\ell + 1}$ without increasing the rank of its neighborhood. That is $ \rank^{(\ell)}(N(B_{\ell + 1})) = \rank^{(\ell)}(N(B_{\ell + 1} \cup i))$, and thus $\InvExp^{(\ell)}(B_{\ell + 1} \cup i) \geq \InvExp^{(\ell)}(B_{\ell + 1}) = \rho$, a contradiction to the maximality of $B_{\ell+1}$.

    \item For contradiction, suppose that $\alg{p}_{\ell + 1} > \alg{p}_\ell$. We will argue that $B_{\ell}$ is not the maximum inverse expansion set in $\mc{M}_{\ell-1}$. Since $B_{\ell}$ and $B_{\ell + 1}$ are disjoint, we have
    \[
        \InvExp^{(\ell-1)}(B_\ell \cup B_{\ell + 1}) = \frac{m(B_\ell) + m(B_{\ell + 1})}{\rank^{(\ell - 1)}(N(B_\ell) \cup N(B_{\ell + 1}))}.
    \]
    Furthermore, recall that $\mc{M}_\ell = \mc{M} / E_{\ell} = \mc{M}_{\ell-1} / \Span_{\mc{M}_{\ell-1}}(N(B_\ell))$, so by the definition of matroid contraction, 
    \[ \rank^{(\ell)}(N(B_{\ell+1})) + \rank^{(\ell - 1)}(N(B_\ell))= \rank^{(\ell - 1)}(N(B_\ell) \cup N(B_{\ell + 1})). \]
    We use this substitution to get our second equality:
    \begin{align*}
        \InvExp^{(\ell-1)}(B_\ell \cup B_{\ell + 1}) &= \frac{m(B_\ell) + m(B_{\ell + 1})}{\rank^{(\ell - 1)}(N(B_\ell) \cup N(B_{\ell + 1}))}  \\
        &= \frac{m(B_\ell) + m(B_{\ell + 1})}{\rank^{(\ell - 1)}(N(B_\ell)) + \rank^{(\ell)}( N(B_{\ell + 1}))}\\
        &> \alg{p}_\ell.
    \end{align*}
    The last strict inequality follows from \Cref{fact:fractions} and our assumption that $\alg{p}_{\ell + 1} > \alg{p}_\ell$. This contradicts that $\alg{p}_\ell$ is the maximum inverse expansion in $\mc{M}_{\ell -1}$. \qedhere
\end{enumerate}
\end{proof}

\subsection{The Gluing Lemma}
\label{sec:gluing-lemma}

Fix an MI market $\Market$, and let $B_1, B_2, \ldots, B_L$ be the sets of buyers peeled off in \Cref{alg:cap}. Fix some $1 \leq r < L$, and define $B_{\leq r} := \bigcup_{\ell \leq r} B_\ell$, and $B_{> r} := B \setminus B_{\leq r}$. Now we define two MI markets:
\begin{itemize}[nosep]
    \item market $\Market^{(1)}$ on the ground set $E^{(1)}:= N(B_{\leq r})$, with matroid $\mc{M}^{(1)} := \mc{M}|_{E^{(1)}}$, and partition matroid $\mc{P}|_{E^{(1)}}$. The buyers are $B_{\leq r}$, and their budgets are the same as in market $\Market$. Let $y^{(1)}, \alg{\alpha}^{(1)}$ be corresponding optimal primal and dual solutions, where $\alg{\alpha}^{(1)}$ is the dual solution from the MI skeleton on market $\Market^{(1)}$.

    \item market $\Market^{(2)}$ on the ground set $E^{(2)}:= E \setminus \Span(E^{(1)})$, with matroid $\mc{M}^{(2)} := \mc{M} / \Span(E^{(1)})$, and partition matroid $\mc{P}|_{E^{(2)}}$. The buyers are $B_{> r}$, and their budgets are the same as in market $\Market$. Let $y^{(2)}, \alg{\alpha}^{(2)}$ be corresponding optimal primal and dual solutions, where $\alg{\alpha}^{(2)}$ is the dual solution from the MI skeleton on market $\Market^{(2)}$.  
\end{itemize}
Additionally, let $\alg{p}_{r+1}$ be the price assigned in step $r+1$ of \Cref{alg:cap} for market $\Market$.
\begin{lem}[Gluing Lemma]
\label{lem:gluing}
Let $y^{+}$ be
\begin{align*}
y^{+}_e :=
    \begin{cases}
      y^{(1)}_e  & \text{if $e \in E^{(1)}$} \\
      0 & \text{if $e \in \Span(E^{(1)}) \setminus E^{(1)}$} \\
      y^{(2)}_e & \text{otherwise}
    \end{cases}
\end{align*}
and $\alpha^{+}$ be
\begin{align*}
\alpha^{+}_S :=
    \begin{cases}
      \alg{\alpha}^{(1)}_{S'} & S = \Span(S')  \text{ for } S' \subsetneq E^{(1)} \text{ with } \alg{\alpha}_{S'}^{(1)} > 0 \\
      \alg{\alpha}^{(1)}_{E^{(1)}} - \alg{p}_{r+1} & S = \Span(E^{(1)})\\
      \alg{\alpha}^{(2)}_{S'}  & S = S' \cup \Span(E^{(1)}) \text{ for } S' \subseteq E^{(2)} \text{ with } \alg{\alpha}_{S'}^{(2)} > 0 \\
      0 & \text{otherwise}.
    \end{cases}
\end{align*}
Then, $y^{+}$ is an optimal solution to the convex program (\ref{eq:EG2}) for market instance
$\Market^+$, identical to market $\Market$. And, in particular, the
prices defined by $\alpha^+$ are identical to those assigned in markets
$\Market^{(1)}$ and $\Market^{(2)}$. 

    Moreover, the lemma still holds if the budgets of buyers in $B_{\leq r}$ are allowed to increase from their original values in market $\Market$, with the increase happening simultaneously in $\Market^{(1)}$ and $\Market^+$.
\end{lem}

\begin{proof}
It suffices to check that $y^+$ and $\alpha^+$ satisfy the KKT conditions for market $\Market^+$. 

\paragraph{Primal Feasibility:} We first verify primal feasibility: for any $S \subseteq E$, we have 
\[
    \sum_{e \in S} y_e^+ = \sum_{e \in S \cap E^{(1)}} y_e^{(1)} + \sum_{e \in S \cap E^{(2)}} y_e^{(2)} \leq \rank_{\mc{M}^{(1)}}(S \cap E^{(1)}) + \rank_{\mc{M}^{(2)}}(S \cap E^{(2)}) \leq \rank(S)
\]
where the first inequality follows from feasibility of $y^{(1)}$ and $y^{(2)}$ in their respective markets, and the second from the definition of matroid contraction. 

\paragraph{Dual Feasibility:}
For all $S \neq \Span(E^{(1)})$, it is clear that $\alpha_S^+ \geq 0$,
so dual feasibility holds. It remains to show for $S =
\Span(E^{(1)})$. We first observe that $\sum_{S' \subseteq E^{(2)}}
\alpha_{S'}^{(2)} = \alg{p}_{r+1}$, since \Cref{alg:cap} on market
$\Market^{(2)}$ is identical to steps $r+1, \ldots, L$ in market
$\Market$ (by construction). Moreover, $\alg{\alpha}^{(1)}_{E^{(1)}}$
is the minimum price of any element in market $\Market^{(1)}$. In the
case that budgets are not increased, $\alg{\alpha}^{(1)}_{E^{(1)}}$ is clearly equal to
$\alg{p}_r$, since the steps \Cref{alg:cap} takes for market $\Market^{(1)}$ are the same as the first $r$ steps on market $\Market$. 

If budgets \emph{have} been increased, then still we have $\alg{\alpha}^{(1)}_{E^{(1)}} \geq \alg{p}_r$, by the monotonicity of prices given in \Cref{lem:competition-monotonicity}. In either case, 
\[
    \alg{\alpha}^{(1)}_{E^{(1)}} \geq \alg{p}_r > \alg{p}_{r+1}.
\]
So $\alpha_{\Span(E^{(1)})}^+ = \alg{\alpha}^{(1)}_{E^{(1)}} - \alg{p}_{r+1} \geq 0$, and dual feasibility is satisfied. 

\paragraph{KKT Stationarity Condition (\cref{eq:KKT1}):}
Let use define
\[ p^+_e := \sum_{S \ni e} \alpha^+_S. \]
and likewise $\alg{p}^{(1)}_e = \sum_{S \ni e} \alg{\alpha}^{(1)}_S$ for $e \in E^{(1)}$ and $\alg{p}^{(2)}_e = \sum_{S \ni e} \alg{\alpha}^{(2)}_S$ for $e \in E^{(2)}$. We first observe that $p^+_e = \alg{p}_e^{(1)}$ for every $e \in E^{(1)}$, and likewise for $e \in E^{(2)}$. This is immediate for every $e \in E^{(2)}$, by the definition of $\alpha^+$. Moreover, since $\sum_{S' \subseteq E^{(2)}} \alpha_{S'}^{(2)} = \alg{p}_{r+1}$, we have that for $e \in E^{(1)}$
\[
    p^+_e = \sum_{S \ni e} \alpha^+_S = p_{r+1} + \sum_{\substack{S \ni e\\S \subseteq \Span(E^{(1)})}} \alpha^+_S = \sum_{\substack{S' \ni e\\S' \subseteq E^{(1)}}} \alg{\alpha}_{S'}^{(1)} = \alg{p}_e^{(1)}
\]
This implies that $\alpha^+$ and $y^+$ satisfy the stationarity condition for $e \in E^{(1)} \cup E^{(2)}$, since $y^{(1)}$, $\alg{\alpha}^{(1)}$, $y^{(2)}$, $\alg{\alpha}^{(2)}$ are optimal and satisfy the stationarity conditions.

Lastly, we must check whether edges $e \in \Span(E^{(1)}) \setminus E^{(1)}$  satisfy the stationarity condition. Say edge $e$ belongs to buyer $i$'s part. By definition, buyer $i$ is a part of market $\Market^{(2)}$. By the the optimality of $y^{(2)}$ and $\alg{\alpha}^{(2)}$ in market $\Market^{(2)}$, we know the price buyer $i$ is purchasing at is precisely equal to 
\[ q_i = \frac{1}{\sum_{e' \in P_i} y^+_e}. \] 
Since buyer $i$ is a part of market $\Market^{(2)}$, we have $q_i \leq \alg{p}_{r + 1}$. By monotonicity of prices, all prices in market $\Market^{(1)}$ are at least $\alg{p}_{r + 1}$. Notice that, by construction, $p_e^+$ has the same price as some element in $E^{(1)}$. Hence, 
\[p^+_e > \alg{p}_{r + 1} \geq q_i = \frac{1}{\sum_{e' \in P_i} y^+_e}. \]
Note that the second part of (\cref{eq:KKT1}) is trivially satisfied since $y^+_e = 0$.

\paragraph{KKT Complementary Slackness (\cref{eq:KKT2}):}
We show that sets with non-zero $\alpha^+_S$ value are tight under $y^+_e$. We go through the cases on what values $\alpha^+_S$ could take. If $\alpha^+_S = \alg{\alpha}^{(1)}_{S'}$  for some $S = \Span(S')$ where $\alg{\alpha}^{(1)}_{S'} > 0$ (case 1), then, 
\[ \sum_{e \in S'} y^{(1)}_e = \rank(S') \]
by complementary slackness on market $\Market^{(1)}$. Therefore, 
\begin{align*}
    \sum_{e \in S} y^+_e &= \sum_{e \in S} y^{(1)}_e + \sum_{e \in S \setminus S'} 0  \\
    &= \rank(S') = \rank(S).
\end{align*}
Similar logic holds for the case $S = \Span(E^{(1)})$ (case 2). Lastly, if $\alpha^+_S = \alg{\alpha}^{(2)}_{S'}$ for some $S = S' \cup \Span(E^{(1)})$, then, since $S'$ is disjoint from $\Span(E^{(1)})$ and $\Span(E^{(1)})$ is a tight set under $y^+_e$, we have 
\begin{align*}
    \sum_{e \in S} y^+_e &= \sum_{e \in S'} y^+_e +  \sum_{e \in \Span(E^{(1)})} y^+_e  \\
    &= \sum_{e \in S'} y^{(2)}_e + \sum_{e \in \Span(E^{(1)})} y^+_e \\
    &= \rank(S') + \rank(E^{(1)}) \\
    &\geq \rank(S' \cup \Span(E^{(1)})) 
\end{align*}
which means $\sum_{e \in S} y^+_e = \rank(S' \cup \Span(E^{(1)}))$ as desired. 
\end{proof}

{\small
\bibliographystyle{alpha}
\bibliography{bib,scheduling}
}

\end{document}